\documentclass[12pt,reqno,sumlimits]{amsart}

\usepackage{amssymb,amscd,amsmath,epsfig}
\usepackage{
  array,
  booktabs, 
  dcolumn,
  rotating,
  shortvrb,
  tabularx,  
  units,
  url
}

\newtheorem{thm}{Theorem}[section]
\newtheorem{cor}[thm]{Corollary}
\newtheorem{lem}[thm]{Lemma}
\newtheorem{prop}[thm]{Proposition}
\theoremstyle{definition}
\newtheorem{defn}{Definition}[section]
\newtheorem{rem}{Remark}[section]

\newcommand{\R}{{\mathbb R}}

\newcommand{\Q}{{\mathbb Q}}
\newcommand{\C}{{\mathbb C}}

\newcommand{\Z}{{\mathbb Z}}

\newcommand{\calA}{{\mathcal A}}
\newcommand{\calB}{{\mathcal B}}

\newcommand{\calD}{{\mathcal D}}
\newcommand{\calE}{{\mathcal E}}
\newcommand{\calF}{{\mathcal F}}
\newcommand{\calG}{{\mathcal G}}

\newcommand{\calM}{{\mathcal M}}

\newcommand{\calL}{{\mathcal L}}

\newcommand{\calV}{{\mathcal V}}

\renewcommand{\to}{\longrightarrow}

\newcommand{\Aut}{\operatorname{Aut}}
\newcommand{\Hom}{\operatorname{Hom}}

\newcommand{\End}{\operatorname{End}}
\newcommand{\ev}{\operatorname{ev}}

\newcommand{\Tr}{\operatorname{Tr}}
\newcommand{\tr}{\operatorname{tr}}

\newsavebox{\savepar}

\numberwithin{equation}{section}

\newcounter{labelflag} 
\newcommand{\labelon}{\setcounter{labelflag}{1}}
\newcommand{\Label}[1]{
                       \ifnum\thelabelflag=1
                          \ifmmode
                             \makebox[0in][l]{\qquad\fbox{\rm#1}}
                          \else
                             \marginpar{\vspace{0.7\baselineskip}
                                        \hspace{-1.1\textwidth}
                                        \fbox{\rm#1}}
                          \fi
                       \fi
                       \label{#1}
                      }
\labelon

\usepackage[top=3cm, bottom=1.8cm, left=2cm, right=2cm,headsep=30pt]{geometry}

\newcommand{\BbP}{{\mathbb P}}

\newcommand{\pdo}{\Psi{\rm DO}}

\newcommand{\fraku}{{\mathfrak u}}
\newcommand{\dvol}{{\rm dvol}}

\newcommand{\dg}{\dot\gamma}

\newcommand{\dir}{\partial\kern-.570em /}
\newcommand{\dire}{\partial\kern-.570em /{}^{\rm eq}}

\newcommand{\idex}{{\rm ind}}

\newcommand{\nlm}{\nabla^{LM}}
\newcommand{\wnlm}{\widetilde\nabla^{LM}}
\newcommand{\nm}{\nabla^{M}}
\newcommand{\wng}{\widetilde\nabla_{\mathfrak u}}
\newcommand{\wog}{\widetilde\Omega_{\mathfrak u}}
\newcommand{\og}{\Omega^M_{\mathfrak u}}

\newcommand{\wolm}{\widetilde\Omega^{LM}}

\newcommand{\bds}{\bar d s}

\newcommand{\nee}{\nabla^E}
\newcommand{\wne}{\widetilde\nabla^\calE}
\newcommand{\woeg}{\widetilde\Omega^\calE_{\mathfrak u}}
\newcommand{\oeg}{\Omega^E_{\mathfrak g}}
\newcommand{\che}{{\rm ch}}

\newcommand{\ints}{\int^{S^1}}
\newcommand{\calas}{\calA^*}
\newcommand{\frakg}{\mathfrak g}
\newcommand{\adP}{{\rm Ad}\ P}
\newcommand\Pad{P^{\rm ad}}
\newcommand{\keras}{{\rm ker}\ d_A^*}

\newcommand{\daga}{d_AG_Ad_A^*}
\newcommand{\cklo}{c_k^{\rm lo}}

\newcommand{\ctlo}{c_2^{\rm lo}}
\newcommand{\PD}{{\rm PD}}
\newcommand{\cgk}{C^\infty_{g,k}(A)}

\newcommand{\trlo}{\Tr^{\rm lo}}
\newcommand{\mgk}{\calM_{g,k}(A)}
\newcommand{\bmgk}{\overline{\calM}_{g,k}(A)}
\newcommand{\mgkm} {\calM_{g,k-1}(A)}

\newcommand{\bv}{[\overline{\calM}_{g,k}(A)]^{\rm virt} }
\newcommand{\bvm}{[\overline{\calM}_{g,k-1}(A)]^{\rm virt} }

\newcommand{\mok}{\calM_{0,k}(A)}
\newcommand{\bmok}{\overline{\calM}_{0,k}(A)}
\newcommand{\mokm} {\calM_{0,k-1}(A)}

\newcommand{\cok}{C^\infty_{0,k}(A)}
\newcommand{\cokm}{C^\infty_{0,k-1}(A)}

\begin{document}

\title[Characteristic classes associated to fibrations]{Equivariant, String and leading order characteristic classes associated to fibrations}
\author[A. Larra\'in-Hubach]{Andr\'es Larra\'in-Hubach}
\address{Department of Mathematics, University of Arizona}
\email{alh@math.arizona.edu}
\author[Y. Maeda]{Yoshiaki Maeda}
\address{Department of Mathematics, Keio University}
\email{ymaeda@math.keio.ac.jp}
\author[S. Rosenberg]{Steven Rosenberg}
\address{Department of Mathematics, Boston University}
\email{sr@math.bu.edu}
\author[F. Torres-Ardila]{Fabi\'an Torres-Ardila}
\address{Department of Curriculum and Instruction, University of Massachusetts Boston}
\email{fabian.torres-ardila@umb.edu}

\maketitle

\section{Introduction}

{\it Added to this version: The results in \S3 are not correct as stated.  The specific error occurs above (3.11), where it is incorrectly stated that $\widetilde\Omega^k_{\fraku}$ takes values in pointwise endomorphisms.  
In fact, as pointed out to us by Tommy McCauley, this operator takes values in first order differential operators.  We have been unable to salvage the results in this section, but feel there are enough worthwhile calculations to leave this section in.}
\bigskip

Infinite rank principal or vector bundles appear frequently in mathematical physics, even before quantization.  For example, string theory involves the tangent bundle to the space of maps ${\rm Maps}(\Sigma^2,M)$ from a Riemann surface to a manifold $M$, while any gauge theory relies on the principal bundle $\calA^* \to \calA^*/\calG$ of irreducible connections over the quotient by the gauge group.  Finally,
formal proofs of the Atiyah-Singer index theorem take place on the free loop space $LM$, and in particular use 
calculations on $TLM.$  As explained below, many of these examples arise from pushing finite rank bundles on the total space of a fibration down to an infinite rank bundle on the base space.

For the correct choice of structure group, these infinite rank bundles can be topologically nontrivial.  As for finite rank bundles, nontriviality is often
detected by   infinite dimensional analogs of the Chern-Weil construction of 
characteristic classes, as  in 
 \cite{h-m-v}, \cite{lrst}, \cite{MRT}, \cite{m-v}, \cite{P-R2},  with a survey in \cite{R}.  The choice of structure 
 group is determined by natural classes of connections on these bundles, which typically take values 
 either in the Lie algebra of a gauge group or a Lie algebra of 
  zeroth order pseudodifferential operators ($\pdo$s).
There are essentially three types of characteristic forms for these connections, one using the 
Wodzicki residue for $\pdo$s, one using the zeroth order or leading order symbol, and one using integration over the fiber.  As shown in \cite{lrst2}, the
characteristic classes for the Wodzicki residue vanish, but nontrivial Wodzicki-Chern-Simons
classes exist \cite{MRT}.  

In this paper, we focus on gauge group connections  and produce examples of nontrivial leading order 
characteristic
classes for some infinite rank bundles associated to loop spaces, Gromov-Witten theory and gauge theory.  While the residue classes
are inherently infinite dimensional objects and difficult to compute, the leading order and string classes for infinite rank bundles on the base space of a fibration
are often related to characteristic classes of the finite rank bundle on the total space. This makes the leading order and string classes more computable.  
In particular, in some cases we can relate  the leading order classes to the string classes associated to 
integration over the fiber.
\medskip

In \S2, we describe the basic setup, which is well known from local proofs of the families index theorem.  To a fibration $Z\to M\to B$ of closed manifolds and a bundle with connection $(E, \nabla)\to M$, one can
associate an infinite rank bundle with connection $(\calE,\nabla') \to M$.  This is a gauge connection if the
fibration is integrable, and we define the associated leading order Chern classes.  Even if the fibration is not integrable, $\calE$ has string classes, which are topological pushdowns of the Chern classes of $E$.  
The leading order and string classes do not live in the same degrees.  Both classes have associated Chern-Simons or transgression forms.  

In \S3, we show that the $S^1$ Atiyah-Singer index theorem can be 
rewritten as an equality involving leading order classes on the loop space $LM$ of a closed manifold $M$ (Thm. \ref{thm2}).  (More precisely, we work with the version of the $S^1$-index theorem called the Kirillov
formula in \cite{BGV}.)
This is an attempt to mimic the formal proofs of the ordinary index theorem on loop space \cite{atiyah}, 
\cite{B}, but differs in significant ways.  In particular, the statement involves integration of
a leading order class over a finite cycle in $LM$, not over all of $LM$, so the nonrigorous
localization step in the formal proof is sidestepped. It should be emphasized that this is 
only a restatement and not
a loop space proof of the index theorem, as the $S^1$-index theorem is used in the
restatement.  Along the way, we construct equivariant characteristic forms on $LM$, such
as the equivariant $\hat A$-genus and  Chern character, which restrict to the corresponding forms
on $M$ sitting inside $LM$ as constant loops (Thm. \ref{thm1}).  It is unclear if the Chern character form we construct is the same as those constructed in \cite{B} and \cite{twz}.

In \S4 we apply similar techniques to the moduli space of pseudoholomorphic maps from
a Riemann surface $\Sigma$ to
a symplectic manifold $M$.  We prove that certain Gromov-Witten invariants and gravitational
descendants 
can be expressed in terms of leading order classes and string classes, and we recover the Dilaton Axiom.   
These techniques work when the
 GW invariants are really given by integrals over the smooth interior of the compactified moduli space, for which we rely on \cite{zi}.  In particular, we have to restrict ourselves to genus zero GW invariants for semipositive manifolds.  The main geometric observation is that the 
fibration of (interiors of) moduli spaces associated to forgetting a marked point is integrable, so that leading order classes are defined. The main results (Thms. \ref{gd}, \ref{pp}) involve a mixture of string and leading order classes.

In \S5 we prove that the real cohomology of a based loop group $\Omega G$, for $G$ compact, is generated by leading order Chern-Simons classes.  This amounts to noting that the cohomology of $G$ is generated by Chern-Simons classes, and then relating these finite rank classes to the leading order classes.
We note that the generators of $\Omega G$ can also be written in terms of string classes, a known result
\cite{h-m-v}, and we specifically relate the string and leading order classes (Thm. \ref{relate}).
Related results are in \cite{hmsv}.

In \S6 we study leading order classes associated to the gauge theory fibration  $\calA^* \to \calA^*/\calG$.  This fibration
has a natural gauge connection  \cite{GP}, \cite{singer}, whose curvature involves nonlocal
Green's operators. 
Leading order classes only deal with the locally defined symbol of these operators, so 
the calculation of these classes is relatively easy.
 In Prop. \ref{nvs}, we show that the canonical representative of  Donaldson's $\nu$-class 
\cite[Ch. V]{dk} in the cohomology of the moduli space $ASD/\calG$ of ASD connections on a 
$4$-manifold is the restriction of a leading order form on all of $\calA^*/\calG.$  Thus the
$\nu$-class gives information on the cohomology of $\calA^*/\calG.$  It is desirable to 
extend this construction to cover the more important $\mu$-classes, but this seems to require a theory of
leading order currents.  We give a preliminary result in this direction.

We would like to thank Michael Murray and Raymond Vozzo for helpful conversations, 
particularly about 
 \S\ref{string}.

\section{Two types of characteristic classes}

Perhaps the simplest type of infinite  rank vector bundles come from fibrations.  Let
$Z\to M\stackrel{\pi}{\to} B$ be a locally trivial fibration, with $Z, M, B$ smooth 
manifolds, and let $E\to M$ be a smooth bundle.  
The pushdown bundle $\calE = \pi_*E$ is a bundle over $B$ with fiber $\Gamma(E|_{\pi^{-1}(b)})$ over $b\in B.$  
To specify the topology of $\calE$, we can choose either a Sobolev class of $H^s$ 
sections for the fibers or the Fr\'echet topology on smooth sections.  

Using the transition functions of $E$, we can check that $\calE$ is a smooth bundle with Banach spaces or Fr\'echet spaces as fibers in these two cases.  For local triviality, take a connection $\nabla$ on $E$, and fix a neighborhood $U$ containing $b$ over which the fibration is trivial.  We can assume that $U$ is filled out by radial curves centered at $b$.  Take a connection for the fibration, {\it i.e.}, a complement to 
the kernel of $\pi_*$ in $TM.$  For $m\in \pi^{-1}(b)$, each radial curve has a unique horizontal lift to a curve in $M$ starting at $m$.  For $s\in \Gamma(E|_{\pi^{-1}(b)})$,
take the $\nabla$-parallel translation of $s$ along each horizontal lift at $m$.  This gives a smooth isomorphism of $\calE_b$ with $\calE_{b'}$ for all $b'\in U.$  

The connection $\nabla$ pushes down to a connection $\pi_*\nabla$ on $\calE$ by
\begin{equation}\label{pdconn}\pi_*\nabla_X(s')(m) = \nabla_{X^h}(\tilde s)(m),
\end{equation}
where $X^h$ is the horizontal lift of $X\in T_bB$ to $T_mM$, and
$s'\in \Gamma(\calE)$ and $\tilde s\in \Gamma(E)$ is defined by 
$\tilde s(m) = s'(\pi(m))(m).$  Thus $\pi_*\nabla$ acts as a first order 
operator on $\calE_b.$
The curvature $\Omega'$ of $\pi_*\nabla$, defined by
\begin{eqnarray*} \Omega'(X,Y) &=& \pi_*\nabla_X\pi_*\nabla_Y - \pi_*\nabla_Y\pi_*\nabla_X - \pi_*\nabla_{[X,Y]}\\
&=& \nabla_{X^h}\nabla_{Y^h} - \nabla_{Y^h}\nabla_{X^h} -\nabla_{[X,Y]^h}\\
&=& \nabla_{X^h}\nabla_{Y^h} - \nabla_{Y^h}\nabla_{X^h} - \nabla_{[X^h, Y^h]} 
+ \left(\nabla_{[X^h, Y^h]} - \nabla_{[X,Y]^h}\right)\\
&=& \Omega(X^h, Y^h) + \left(\nabla_{[X^h, Y^h]} - \nabla_{[X,Y]^h}\right),
\end{eqnarray*} 
satisfies $\Omega'(X,Y) (s')(m) = \Omega(X^h, Y^h)(\tilde s)(m)$ iff
$[X,Y]^h = [X^h, Y^h]$, {\it i.e.}, iff the connection for the fibration has vanishing curvature
\cite[p.~20]{BGV}.
$\Omega(X^h, Y^h)$ is a zeroth order or multiplication operator, so in
 general, $ \nabla_{[X^h, Y^h]} - \nabla_{[X,Y]^h}$ and hence $\Omega'$
acts on  $\calE_b$ as a first order differential operator: $\Omega\in \Lambda^2(B,\calD^1)$ in the obvious notation.

For a finite rank bundle $F\to B$ with connection, the curvature lies in $\Lambda^2(B,
\End(F))$, and
Chern classes are built from the usual matrix trace tr on $\End(F).$  There are no known 
nontrivial traces on $\calD^1$, as the Wodzicki residue vanishes on differential operators.  However, if the connection on the fibration is integrable, then a 
 natural trace $\trlo$ on multiplication operators can be built from the matrix trace as follows.
 Take a Riemannian metric on $M$, giving a volume form 
$\dvol_b$  on each fiber $M_b.$\footnote{If we choose this metric initially, we can take the horizontal distribution to be the orthogonal complement to Ker $\pi_*.$}
 For $\eta\in \Lambda^k(B,\End(\calE))$ locally of the form $\eta = \sum_i a_{i}\omega^i\otimes A^i$
 with $\omega^i \in \Lambda^k(B), A^i\in \End(\calE)$ 
 \begin{equation}\label{trlo}
 \trlo(\eta) = \sum_i a_i\omega^i\int_{M_b} \tr (A^i)\dvol_b\in \Lambda^k(B).
 \end{equation}
 In \cite{P-R2}, this trace is called the leading order trace, as it extends to bundles whose transition functions are zeroth order pseudodifferential operators.  
 
 \begin{defn} {\it The leading order Chern classes of $\calE$ are 
 $$\cklo(\calE) = \left[ \trlo(\Omega^k)\right]\in H^{2k}(B, \C),$$
 where the brackets denote the de Rham cohomology class.  The leading order
 Chern character of $\calE$ is}
 $$ch^{\rm lo}(\calE) = \left[\trlo(\exp(\Omega))\right]\in H^{\rm ev}(B, \C).$$
 \end{defn}
 \medskip
 
 For 
 $k=0$, $\Tr(AB) = \Tr(BA)$, so the usual proof that the Chern form $\Tr(\Omega^k)$
 is closed with de Rham class independent of choice of connection carries over.  

The integral in $\Tr(\eta)$ is an averaging of the endomorphism and leaves the degree of $\omega^i$ unchanged.
In contrast, we can integrate $\omega^i$ over the fiber as well, which we denote by
$\int_Z\omega^i$ or $\pi_*\omega^i.$  This leads to a second type of Chern class, called string classes  
\cite{m-v} or caloron classes \cite{h-m-v}.  Let $z = {\rm dim}\ Z.$

\begin{defn} \label{2.2}{\it The string classes of $\calE$ are}
$$c^{\rm str}_k(\calE) = \left[ \pi_*\Tr(\Omega^k)\right]  = \pi_* \left[ \Tr(\Omega^k)\right]\in H^{2k -z}(B, \C).$$
{\it The string Chern character of $\calE$ is}
$$ch^{\rm str}(\calE) = \left[\pi_*\Tr( \exp(\Omega))\right] \in H^{{\rm ev} -z}(B,\C).$$
\end{defn}

The  $\pi_*$ outside the brackets denotes the induced
pushforward on de Rham cohomology.  Thus string classes satisfy a naturality condition:
\begin{equation}\label{1.1} 
c^{\rm str}_k(\pi_*E) = \pi_* c_k(E).
\end{equation}

 The string classes are 
 the topological pushforward, and so 
can be defined for any coefficient ring, most easily
if the total space $M$ of the fibration is compact.
  Specifically, for Poincar\'e duality $\PD_M:H_*(M,\Z)\to H^{{\rm dim}\ M-*}(M,\Z)$, we have 
\begin{equation}\label{PD}c^{\rm str}_k(\pi_*E) = \PD_B\circ \pi_*\circ \PD_M^{-1} c_k(E),
\end{equation}
where $\pi_*$ on the right hand side is the usual homology pushforward. Thus the string classes are novel only in that they are identified with characteristic classes of infinite rank bundles.   

In contrast,
the leading order classes have no obvious interpretation for $\Z$ coefficients.  However, the leading order forms for a fibration contain more information than the string forms: the string forms  average only terms with 
all fiber variables and discard the rest, while the leading order forms average all terms as in (\ref{trlo}).

Both classes have associated Chern-Simons classes; see \S5.  As with ordinary CS classes, the leading order and string CS classes are geometric objects.

\section{Equivariant characteristic classes on loop spaces}

{\it The results in this section are correct until (3.11); see the footnote there.}
\bigskip

In this section we construct equivariant characteristic classes on $LM$
and relate them to the corresponding characteristic classes on $M$. We use these constructions to restate
 the $S^1$ index theorem in terms of data on $LM$. 

Let $M$ be a closed, oriented, Riemannian $n$-manifold. Fix a 
parameter $s \gg 0$, and let $LM = L^{(s)}M$ be the space of maps $f:S^1\to M$ of Sobolev class 
$s\gg 0$.  $LM$ is a Banach manifold.  The space $L^{(\infty)}M$ of smooth loops is only a 
Fr\'echet manifold; the techniques of this section work in this case as well.

The tangent space $T_\gamma LM$ at a loop $\gamma$ consists of vector fields along $\gamma$, 
{\it i.e.,} sections of $\gamma^*TM\to S^1$ of either Sobolev class $s$ or smooth.  In sheaf theory terms, $TL^{(\infty)}M = \pi_*\ev^*TM$, for\\
 $\ev: LM\times S^1\to M$ the evaluation map
$\ev(\gamma,\theta) = \gamma(\theta)$ and $\pi:LM\times S^1\to LM$ the projection. The following diagram encapsulates the setup.

\begin{equation}\label{encap}\begin{CD} 
@.\textrm{ev}^*TM @>>> TM\\
@.@VVV   @VVV\\
@.LM\times S^1 @>{\rm ev}>> M\\
@.@V\pi VV @.\\
TLM = \pi_*\ev^*TM @>>> LM  @.
\end{CD}\ \ \ \ \ \ \ \ \ \ \ \ \ \ \ \ \ \ \ \ \ \ \ \ \ \ \ \ \ \ \ \end{equation}
We remark that since $T_\gamma LM$ is noncanonically isomorphic to the trivial bundle
$\underline {\R}^n = S^1\times \R^n\to S^1$,   the structure group of $TLM$ is the gauge group of $\underline {\R}^n.$

$LM$ has the $L^2$ Riemannian metric
$$\langle X, Y\rangle_\gamma = \frac{1}{2\pi}\int_{S^1}\langle X(\theta), Y(\theta)\rangle_{\gamma(\theta)} d\theta.$$
Let $\nlm$ be the $L^2$ connection on $LM$ associated to the Levi-Civita
connection $\nm$ on $M$.  This is given by ``pulling  back and pushing down $\nm$ 
to $LM$."  
To define this carefully, particularly  at self-intersections
of $\gamma$, pick $X\in T_\gamma LM$, $Y\in \Gamma(TLM),$ define
$\tilde\gamma:(-\epsilon,\epsilon)\times S^1\to M$ by
$\tilde\gamma(t,s) = \exp_{\gamma(s)}tX(s),$ 
and define the vector field $Y(t,s)$ on $(-\epsilon,\epsilon)\times S^1$ by 
$Y(t,s)  = Y_{\tilde\gamma(t,s)}.$  Then
\begin{equation}\label{l2lc} 
(\nlm_XY)_\gamma(s) = [(\ev^*\nabla^M)_{(X,0)}\widetilde Y]_{(\gamma,s)},
\end{equation}
where $\widetilde Y\in \Gamma(\ev^*TM)$ is $\widetilde Y_{(\gamma,s)} = Y_{\gamma(s)}.$ 

Similarly, if $(E,\nabla^E,h)\to M$ is a hermitian vector bundle with connection, we can form $\calE = \pi_*\ev^*E\to LM$ with fiber
$\calE_\gamma = \Gamma(\gamma^*E\to S^1)$ in the Fr\'echet case.  The structure group of $\calE$ is
the gauge group of $\underline {\R}^k\to S^1$, $k = {\rm rk}(E).$  
$\calE$ has an $L^2$ hermitian 
metric $\langle e_1, e_2\rangle_\gamma = (2\pi)^{-1}\int_M h(e_1, e_2) d\theta$.  As above, $\nabla^E$  
pushes down to a hermitian connection $\nabla^\calE$ on $\calE.$


$LM$ has a canonical $S^1$ action $k_s:LM\to LM, s\in [0,2\pi],$ given by rotation of loops: $k_s(\gamma)(\theta) =
\gamma(s+\theta).$  $k_s$ is an isometry of $LM$.  
$\calE$ is an equivariant bundle for this action:
$$\begin{CD} \calE@>k'_s>>\calE\\
@VVV @VVV\\
LM@>>k_s> LM
\end{CD}$$
where $k_s':\calE_\gamma\to \calE_{k_s\gamma}, k_s'(e)(\theta) = e(s+\theta).$

As pointed out in \cite{B}, $\nlm$ is not an $S^1$-invariant connection.
Recall that an invariant connection would satisfy \cite[p. 26]{BGV}
\begin{equation}\label{zero}\nlm k_{s,*} =k_{s,*}\nlm.
\end{equation}
More explicitly, let $W = W_\gamma (s)$ be a vector field on $LM$ and $X\in
T_\gamma LM.$  Then (\ref{zero}) means
\begin{equation}\label{zero1} ({\nlm}_X k_{s,*}[W_{k_{-s}(\gamma)}]) (s_0) = k_{s,*}\left(
 [ {\nlm}_{k_{-s,*}X}W]_{k_{-s}(\gamma)}\right)(s_0)
\end{equation}
at $\gamma.$
To compute $k_{s,*}:T_{k_{-s}(\gamma)}LM \to T_\gamma LM$, set
$\gamma_t(s) = \exp(tW(s))$.  Then
\begin{equation}\label{zero0}
(k_{s,*}W_\gamma)(s_0) =
\frac{d}{dt}\biggl|_{t=0} (k_s(\gamma_t))(s_0) =
\frac{d}{dt}\biggl|_{t=0}\exp(tW(s+s_0)) = W_\gamma(s+s_0) = W_{k_s(\gamma)}(s_0),
\end{equation}
which is indeed a vector field along $k_s(\gamma).$
By (\ref{l2lc}),
the left hand side of (\ref{zero1}) equals
$$[(\ev^*{\nm}_{(X,0)} )  \widetilde W]_{(\gamma, s_0)}\in \ev^*TM\bigl|_{(\gamma, s_0)}
= T_{\gamma(s_0)}M,$$
while the right hand side equals
$$k_{s,*}\left[(\ev^*\nm)_{(k_{-s,*}X,0)}\widetilde W\right]_{(k_{-s}(\gamma),s_0)}
\in k_{s,*}\left(\ev^*TM \bigl|_{(k_{-s}(\gamma),s_0)} \right)= T_{\gamma(s_0)}M.$$
Even though both sides of (\ref{zero1}) are vectors in
$T_{\gamma(s_0)}M$, the two sides  differentiate $W$ at the different points
$\gamma(s_0), \gamma(-s+s_0)$, and so $\nlm$ is not $S^1$-invariant.  


As in \cite[(1.10)]{BGV}, we can average $\nlm$ over the action to
produce 
\begin{equation}\label{wnlm}
\wnlm = \int_0^{2\pi}\left(k_s^{T^*
LM\otimes TLM}\right)^{-1} \nlm
k_{s,*}\ \bds,
\end{equation} 
where 
$(k_s^{T^*LM\otimes TLM})^{-1}\nabla^{LM}_XY = \nlm_{k_{-s,*X}} k_{-s,*}Y$ 
and  $\bds = \frac{1}{2\pi} ds.$
$\wnlm$ is $S^1$-equivariant, since (simplifying the notation)
\begin{eqnarray*}k_{s_0}\wnlm &=& k_{s_0}\left(\int_0^{2\pi}k_s^{-1} \nlm
k_s\ \bds\right) = \int_0^{2\pi}k_{s_0-s}
\nlm k_s\ \bds\\
& =& \int_0^{2\pi}k_{-s}
\nlm k_{s+s_0}\ \bds
=  \wnlm k_{s_0}.\end{eqnarray*} 
We can similarly average $\nabla^\calE$ to obtain an $S^1$-invariant connection $\widetilde\nabla^\calE.$

We now follow \cite[\S7.1]{BGV}.  Let
$ (\C[\fraku]\otimes \Lambda^*(LM))^{S^1}$
be the space of equivariant forms on $LM$ with values 
polynomials on
$\fraku = \fraku(1)$, the Lie algebra of $S^1.$  Equivalently, this is the space of 
equivariant polynomial maps from 
$\fraku$ to $\Lambda^*(LM)$.
For deg$(u) = 2$, this space becomes a complex for the degree one equivariant differential
$(d_\frakg \alpha)(u)= d(\alpha(u))- \iota_u \alpha(u), u\in \fraku$, where $\iota_u$ is the interior product 
of the vector field on $LM$ associated to $u$.  In particular, if $u = \partial_\theta$ in the usual notation, 
then $\iota_u = \iota_{\dot\gamma}$ at the loop $\gamma$ (so $\dot\gamma\in T_\gamma LM)$).
The cohomology of this complex  at $u=0$ is the Cartan model for the equivariant cohomology $H_{S^1}(LM, \C)$.

$\widetilde\nabla^\calE$ has the associated so-called equivariant connection 
$\widetilde\nabla^\calE_\fraku$ acting on 
$(\C[u]\otimes \Lambda^*(LM,\calE))^{S^1}$:
$$(\widetilde\nabla^\calE_\fraku\alpha)(X) = (\widetilde\nabla^\calE-\iota_X)(\alpha(X)),\ \ X\in \fraku.$$
 With $\calE$ and $\alpha$ understood, we also denote the left hand side of this equation by $\wng(X).$ The 
 equivariant curvature is by definition
$$\wog^\calE = \wng^2(X) + L_{X}^{\calE}.$$
 Here $L_X^\calE$ is the Lie derivative along the vector field on the $S^1$-manifold $\calE$ determined by
 $X\in\fraku.$ 
 By \cite[Prop. 7.4]{BGV}, $\wog^\calE \in (\C[u]\otimes\Lambda^*(LM, \End(\calE)))^{S^1}$.  

Assume that $M$ has an $S^1$ action $a:S^1\times M\to M$. The case of the trivial 
action $a(\theta, m) = m$ is not uninteresting.
By averaging the metric over $S^1$, we may assume that the action is via isometries.
 Let $a_s:M\to M$
be $a_s(m) = a(s,m).$ The action induces
an embedding 
$$a':M\to LM, \ \  a'(m) = (s\mapsto a(s,m)).$$
The following diagram commutes:
\begin{equation}\label{diagram}\begin{CD}M @>a_s>> M\\
@V{a'}VV  @VV{a'}V\\
LM @>>{k_s}> LM
\end{CD}
\end{equation}
Let $Y$ be the vector field for the flow $\{a_s\}$ on $M$, {\it i.e.}, $Y$ is the vector field corresponding to 
$\partial_\theta\in \fraku.$  Since $\{k_s\}$ is the flow of $\dot\gamma$ on $LM$,
 it follows from (\ref{diagram}) that
for a vector field $V$ on $M$,
\begin{equation}\label{nine} L_{\dot\gamma}(a'_*V) = \frac{d}{ds}\biggl|_{s=0}(k_{-s}\circ a')_* V
= \frac{d}{ds}\biggl|_{s=0}(a'\circ a_{-s})_*V 
=  a'_*L_YV,
\end{equation}

For $V_0\in T_mM$, we have 
\begin{equation}\label{ab}(a'_*V_0)(s) = (a_s)_*(V_0).
\end{equation}
 From now on, we 
denote $a'$ just by $a$, so $(a_*V_0)(s) = (a_s)_*(V_0)$.
Let $i:M\to LM$ be the isometric embedding taking a point to a constant loop.  On
Fix, the fixed point set of $a$, we have $a = i$  and $ia_s = k_s i.$
Let $T\subset TLM$ be the rank $n$ subbundle of $TLM$ of ``rotated vectors": the fiber is $T_\gamma =
\{s\mapsto a_{s,*}(V_0), V_0\in
T_{\gamma(0)}M\}$.  Thus $V\in\Gamma(TM)$ implies $a_*V\in\Gamma(T).$  Clearly $i^*T\simeq TM$
and $\beta:a^*T\stackrel{\simeq}{\to}  TM$ via $\beta(s\mapsto
k_{s,*}(V_0))= V_0.$    $\beta$ induces an isomorphism
$\Lambda^*(M, a^*T)\stackrel{\simeq}{\to} \Lambda^*(M,TM).$


 
 We will  need the analog of (\ref{nine}) for the Levi-Civita connection.
 \begin{lem}  Under the isomorphism $\beta: a^*T\simeq TM$, we have
 \begin{equation}\label{lem1}
 a^*(\nlm_\cdot\dot\gamma) = \nm_\cdot Y\in \Lambda^1(M,TM).
 \end{equation} 
 \end{lem}
 \begin{proof}  Extend a fixed $V_0\in T_{m_0}M$ to a vector field $V$ on $M$, so $a_*V$ is an extension of
 $a_*V_0\in T_{a(m_0)}LM$ to a vector field on $a(M) \subset LM.$  Then
 \begin{eqnarray*} a^*(\nlm_\cdot\dot\gamma)(V_0) &=& \nlm_{a_*V_0}\dot\gamma
 = \nlm_{\dot\gamma}a_*V +[a_*V,\dot\gamma]
 = \nlm_{\dot\gamma}a_*V +[a_*V, a_*Y]\\
 &=& \nlm_{\dot\gamma}a_*V +a_*[V,Y]
=
    \nlm_{\dot\gamma}a_*V -a_*L_YV.
 \end{eqnarray*}
In local coordinates,  
$$(\nlm_{\dg} a_*V)^i = (s\mapsto \dot \gamma(a_{s,*}V)^i + \Gamma^i_{jk}\dg^j(a_{s,*}V)^k)
=(s\mapsto \nm_Ya_{s,*}V = a_{s,*}\nm_YV),$$
since  $Y = a_{s,*}Y$ is the velocity vector field for the orbit $\gamma$ and
$a$ acts via isometries.    Since 
$a_{*}L_YV = (s\mapsto a_{s,*}L_YV)$, we have
$$a^*(\nlm_\cdot\dot\gamma)(V_0)= (s\mapsto a_{s_*}(\nm_YV-L_YV)) = (s\mapsto a_{s,*}(\nm
_VY))\in T|_\gamma = a^*T|_{m_0}$$
Thus using $\beta:a^*T\stackrel{\simeq}{\to} TM$, we have
$a^*(\nlm_\cdot \dg) = \nm_\cdot Y.$
\end{proof}

 \medskip

 We now focus on the Riemannian case with $E = TM.$
 
 Let 
 $\Omega_\fraku = Ма\og$ be the 
equivariant curvature of the $S^1$-invariant Levi-Civita connection $\nm$, and let $\widetilde \Omega_\fraku$
be the equivariant curvature of $\widetilde\nabla^{LM}.$   Since $\widetilde \Omega^k_\fraku
\in (\C[u]\otimes\Lambda^*(LM, \End(TLM)))^{S^1}$ takes values in pointwise 
endomorphisms,\footnote{This statement is incorrect, as a calculation by T. McCauley shows that  $\widetilde \Omega^k_\fraku$ takes values in first order differential operators.}
its
powers have a leading order trace
 \begin{equation}\label{lot}\Tr(\wog^k) = \int_{S^1} \tr  (\widetilde\Omega_\fraku^k(s)) \bar d s\in
 (\C[u]\otimes\Lambda^*(LM))^{S^1}.
 \end{equation}

\begin{lem} \label{lem2}$a^*\Tr(\widetilde \Omega^k_\fraku) = \tr((\og)^k)$.
\end{lem}

Since the curvature form is skew-symmetric, both sides are zero if $k$ is odd.
\medskip

\begin{proof} 

 Denote the
  orbit $a(m)$ by $\gamma.$
Because the action is via isometries, we have  \cite[p. 209-210]{BGV},
\begin{equation}\label{eqcurv} \wog = 
\wolm -\wnlm_\cdot {\dot\gamma}, 
\end{equation}
where $\wolm$ is the curvature of $\wnlm.$ On the right hand side of (\ref{eqcurv}),
$\wolm\in\Lambda^2(LM, \Hom(TLM))$, and  $\wnlm_\cdot {\dot\gamma}\in
\Lambda^0(LM, \Hom(TLM))$.  
Thus for $Y_1, Y_2, Z\in
T_\gamma LM$, $\wog(Y_1, Y_2)Z = 
\wolm(Y_1, Y_2)Z - \wnlm_Z{\dot\gamma}
\in T_\gamma LM$.  Because of the pointwise nature of the $L^2$ connection, we have
\begin{equation}\label{insert0}
\Omega^{LM}(a_*Y_1, a_*Y_2)(s) = \Omega^M(a_*Y_1(s), a_*Y_2(s) ),\ 
(\nabla^{LM}_{a_*Z}\dot\gamma) (s) = \nabla^M_{a_*Z(s)}\dot\gamma,
\end{equation}
and for $Z(0)\in T_mM$,
\begin{eqnarray}\label{insert}
\Omega^M(a_*Y_1(s), a_*Y_2(s) )a_{s,*}Z &=& a_{s,*}[\Omega^M(a_{s,*}(Y_1(0)), 
a_{s,*}(Y_2(0))Z(0)] = a_{s,*}[\Omega^M((Y_1(0), Y_2(0)Z(0)],\nonumber\\
\nabla^M _{a_{s,*}Z(0)}\dot\gamma &=& a_{s,*}\nabla^M_{Z(0)}\dot\gamma.
\end{eqnarray}

Let $\{e_i \}$ be an orthonormal frame of $T_mM$, so $\{a_{s,*}e_i  = e_i(s)\}$ is an orthonormal frame at $T_{a(m,s)}M.$
It follows from (\ref{eqcurv}) --  (\ref{insert}) that 
$$\tr((a^*(\Omega^{LM} - \nabla^{LM}_\cdot\dot\gamma)^k(Y_1,\ldots,Y_{2k}))  = 
\langle (\Omega^{LM} - \nabla^{LM}_\cdot\dot\gamma)^k(a_{s,*}Y_1,\ldots, a_{s,*}Y_{2k})(e_i (s)), 
e_i (s)\rangle_{a(m)(s)}$$
is independent of $s$.
Trivially averaging this expression over $s$ to obtain $\widetilde \Omega_\fraku$, we see that
$\widetilde \Omega_\fraku$ acts pointwise in $s$. Therefore
\begin{eqnarray}\label{ergo}\tr((a^*\widetilde\Omega_\fraku)^k)(Y_1,\ldots,Y_{2k})_m &=& 
\langle (\widetilde\Omega^{LM} - \widetilde\nabla^{LM}_\cdot\dot\gamma)^k(a_{s,*}Y_1,\ldots, a_{s,*}Y_{2k})(e_i (s)), 
e_i (s)\rangle_{a(m)(s)}\biggl|_{s=0}\nonumber \\
&=& \langle (\Omega^{M} - \nabla^{M}_\cdot Y)^k(Y_1,\ldots, Y_{2k})
(e_i),e_i\rangle_m\\
&=& \langle \Omega_\fraku^k(Y_1,\ldots, Y_{2k})
(e_i),e_i\rangle_m\nonumber\\
&=& \tr( \Omega_\fraku^k)(Y_1,\ldots,Y_{2k})_m.\nonumber
\end{eqnarray}
where the second line 
follows from Lemma \ref{lem1} and (\ref{insert0}) (and noting that $\nabla^M$ is already equivariant).

There is one final average in (\ref{lot}) to obtain
$$a^*\Tr(\widetilde\Omega^k_\fraku)_m= \int_{S^1}\tr((a^*[(\widetilde 
\Omega_\fraku)_{a(m)(s)}])^k)\bar d s.$$ 
However, (\ref{ergo}) shows that 
$$a^*\Tr(\widetilde\Omega^k_\fraku)_m = \int_{S^1} \tr((a^*[(\widetilde 
\Omega_\fraku)_{a(m)(0)}])^k)\bar d s = \tr((a^*[(\widetilde 
\Omega_\fraku)_{s}])^k) = \tr(\Omega_\fraku^k)_m.$$
\end{proof}

By this Lemma, we can extend the $\hat A$-polynomial as a characteristic form in
the curvature on $M$ to
an equivariantly closed form on $LM$.
The $\hat A$-polynomial of a curvature form $\Omega$ 
can be expressed as a polymonial in $\tr(\Omega^{2k}).$  In particular, $\hat
A(\wog)$ is defined using the leading order trace Tr.

Let $T|_M$ be the restriction of $T\subset TLM$ to the constant loops
$i(M)\subset LM.$

\begin{thm}\label{thm1} (i) $a^*\hat A(\wog) = \hat A(\og)$.   

(ii) $\hat A(\wog)$ is an
  equivariant extension of the $\hat A$-form on constant loops, i.e.  $\hat
  A(\wog)  = \hat A(\Omega^M)$ when restricted to  $T|_M.$
\end{thm}

As above, these equalities use the isomorphisms $i^*T\simeq a^*T\simeq TM$ to
identify $  a^*({\nlm}_\cdot {\dot\gamma}), a^*\Omega^{LM}$, with 
$\nabla^M_\cdot Y, \Omega^M$, respectively.  (i) follows immediately from the Lemma, standard 
Chern-Weil theory,  and the fact that the equivariant curvature satisfies the Bianchi identity
\cite[\S7.1]{BGV}. (ii)
follows from $\dot\gamma = 0$
on $i(M).$

\begin{rem}  (i) This discussion extends to equivariant hermitian
bundles $E\to M$ and their
  loopifications $\calE = \pi_*\ev^*E\to LM.$  Namely, we can take an invariant
  connection $\nee$ on $E$, form the $L^2$/pointwise connection
  $\nabla^\calE$
  on $\calE$, average it to $\wne$, and prove that the corresponding
  equivariant   curvatures satisfy
$$a^*(\woeg) = \oeg.$$
In particular, the Chern characters satisfy 
$$a^*\che(\woeg) = \che(\oeg),$$
and the Chern character restricts to the ordinary Chern character on the constant loops.  See \cite{twz} for 
an alternative construction of an equivariant Chern character on $LM$.  By the localization formula for equivariantly closed forms, the top degree rational equivariant cohomology classes of these forms are determined by their values on constant loops, where they agree.  It remains to be seen if the actual forms agree.

(ii)   If $\Omega_\fraku^{LM} = \Omega^{LM} - \nabla^{LM}_\cdot\dot\gamma$ is built from the 
$L^2$ connection and curvature on $LM$, the proof of Lemma \ref{lem2} (without the final $S^1$ average) 
implies that
$a^*\Tr(\Omega^k_\fraku) = \tr((\Omega^M_\fraku)^k).$  Thus $a^*\hat A(\Omega_\fraku^{LM}) = \hat 
A(\Omega_\fraku^M).$  This is somewhat more natural than Theorem \ref{thm1}(i), but 
$\Omega_\fraku^{LM}$ and hence $\hat A(\Omega_\fraku^{LM})$ are not equivariantly closed.
\end{rem}

\medskip

Recall that $a:M\to LM$ by abuse of notation.

\begin{defn} {\it For an $S^1$ action $a:S^1\times M\to M$, set
  $[a] = a_*[M^n]\in H_n(LM,\Z)$.}
\end{defn}

Since $a:M\to LM$ is injective, we also denote its image by $[a]$, an 
$n$-dimensional submanifold of $LM$.  

We now review the $S^1$-index theorem.  Since $S^1$ is assumed to act on $(E,\nabla^E)$ covering its action via isometries on $M$, the kernel
and cokernal of $\dir_\nee$ are representations of $S^1$. The $S^1$-index of
$\dir_\nee$ is
the corresponding element of the representation ring $R(S^1)$:
$${\rm ind}_{S^1}(\dir_{\nee}) = \sum(a_k^+ - a_-^k)u^k\in \Z[u,u^{-1}] = R(S^1),$$
where $u^k$ denotes the representation $e^{i\theta}\mapsto e^{ik\theta}$ of $S^1$ on $\C$, and 
$a_\pm^k$ are the multiplicities of of $u^k$ in the kernel and cokernel of $\dir_{\nee}.$
For a general compact group $G$, the  $G$-index theorem identifies this 
analytically defined element of $R(G)$ with a topologically defined element.  This is difficult to compute in general, but there is a formula to compute  
the character of the action of a fixed $g\in G$ on the index space
$[{\rm ker}\ \dir_{\nee}] - [{\rm coker} \ \dir_{\nee}]\in K_G({\rm pt}) = R(G)$; this is 
often also called the $G$-index theorem.  
For $g = e^{i\theta}\in S^1$, the character is just 
${\rm ind}_{S^1}( e^{ik\theta},\dir_\nee) = \sum(a_k^+ - a_-^k)e^{ik\theta}.$  

The Atiyah-Segal-Singer fixed point formula for the $S^1$-index computes
${\rm ind}_{S^1}(\dir_{\nee}) $ in terms of data on the fixed point set of a particular $e^{i\theta}$ \cite[Thm. 6.16]{BGV}.  Using the localization theorem for integration in equivariant cohomology \cite[Thm. 7.13 ]{BGV}, this can be rewritten as 
\begin{equation}\label{gind} {\rm ind}_{S^1}( e^{-ik\theta}, \dir_{\nee})  =
 (2\pi i)^{-{\rm dim}(M)/2}\int_M 
\hat A_{\fraku}(\theta, \Omega_\fraku) {\rm ch}(\theta,\Omega_{\fraku}^E),
\end{equation}
 \cite[Thm. 8.2]{BGV}. Here $\hat A_{\fraku}(\theta, \Omega_\fraku) = \hat A(\Omega_\fraku)(\theta)
\in\Lambda^*(M)$ is the evaluation of $ \hat A(\Omega_\fraku)\in (\C[u]\otimes\Lambda^*M)^{S^1}$ at $\theta\in \fraku(1).$  (For general compact groups $G$, this theorem only holds for group elements close to the identity.  For the $S^1$-index theorem, both sides are analytic for $\theta$ small, and so the equality extends to all $\theta.$)
For notational ease, we rewrite (\ref{gind}) as
\begin{equation}\label{ease}\overline{\rm ind}_{S^1}( \dir_{\nee})  =
 (2\pi i)^{-{\rm dim}(M)/2}\int_M^{S^1} 
\hat A_{\fraku}(\Omega_\fraku) {\rm ch}(\Omega_{\fraku}^E),
\end{equation}
with the left hand side evaluated at $e^{-ik\theta}$ and the right hand side evaluated at $\theta.$

We can now restate the $S^1$-index theorem for the Dirac operator as a result involving the equivariant curvature of $LM$.  Unlike the usual statement, in this version the action information is contained precisely in $[a]$, while the integrand depends only on the (action-compatible)
Riemannian metric on $M$.  

\begin{thm}  \label{thm2} Let $M$ be a spin manifold with an isometric $S^1$-action, and let $E$ be an equivariant hermitian bundle with connection $\nabla^E$ over $M$.  Then
$$\overline{\rm ind}_{S^1}(\dir_{\nee})  = (2\pi i)^{-{\rm dim}(M)/2} \ints_{[a]}\hat A(\wog)\che(\woeg).$$
\end{thm}  

\begin{proof}  By Thm. \ref{thm1}, 
\begin{eqnarray*} \ints_{[a]} \hat A(\wog)\che(\woeg) &=& \ints_{a_*[M]}\hat A(\wog)\che(\woeg) = \ints_{[M]} a^*(\hat
A(\wog)\che(\woeg) )
=  \ints_M \hat A(\Omega_{\mathfrak u}){\rm ch}(\Omega^E_{\mathfrak u}).
\end{eqnarray*}
The $S^1$-index theorem in the form (\ref{ease}) finishes the proof.
\end{proof}

\begin{rem}  (i) By the localization theorem, we have
$$
\overline\idex_{S^1}(\dir_{\nee}) = \int_{[a]}^{S^1}\hat A(\wog)\che(\woeg)
= \int_{[a]\cap M}^{S^1} \frac{\hat A(\wog)\che(\woeg)}
{\chi_\fraku(\nu_a^{LM})},
$$
where $\nu_a^{LM}$ is the normal bundle of the fixed point set ${\rm Fix}(a) = 
[a]\cap M$ in $[a]$, and $\chi_\fraku$ is the equivariant Euler form.

(ii) If we do not assume that the action $a$ is via isometries, then as in \cite{twz}, we should replace the rotational action $k_s$ on $LM$ by parallel translation along loops.  The averaging procedure again produces an equivariantly closed form extending a given characteristic form on $M$.
 However, Thm. \ref{thm2} does not extend.

(iii)  The integrand in Thm. \ref{thm2} is not really independent of the action, since it depends on the 
action-dependent metric. However, we
can  push this metric dependence out 
of the integrand as follows.  Let $\calB$ be
the space of metrics on $M$.  $\calB$ comes with a natural Riemannian metric $g^\calB$, the so-called $L^2$ metric, given at $T_{g_0}\calB$ by
$$g^\calB (X,Y) = \int_M g_0^{ab}g_0^{cd}X_{ac}Y_{bd}
\dvol_{g_0}.$$
Thus $LM\times \calB$ has a metric $h$ which at $(\gamma, g_0)$ is the product metric of the 
 $L^2$ metric on $T_\gamma LM$ 
determined by $g_0$ and $g^\calB$ on $T_{g_0}\calB.$ 
This is not a global product metric, but it is not difficult to compute the Levi-Civita connection and curvature $F$ of 
$h$.  We extend the rotational action on $LM$ trivially to $LM\times \calB$, so one
obtains an equivariant curvature $\widetilde F_{\mathfrak u}$ .  One directly 
computes that $\widetilde F^{(1)}_{\mathfrak u} = P^{TLM}\widetilde F_\fraku P^{TLM}$ equals
$\wog$, where $P^{TLM}$ is the 
$h$-orthogonal projection of $T(LM\times \calB)$ to $LM.$   One obtains

\begin{prop} Let
$i_{g_0}:LM\times\{g_0\}\to LM\times\calB$ be the inclusion.  
 If $a$ is a $g_0$-invariant $S^1$ action on $M$,  then
$$\overline\idex_{S^1}(\dir) = \int^{S^1}_{i_{g_0,*}[a]} \hat A(\widetilde F_\fraku ^{(1)}).$$
\end{prop} 

Thus the integrand is a universal form on $LM\times \calB$, and the choice of action and compatible metric are encoded in the cycle of integration.
\end{rem}

\section{Gromov-Witten theory}
In this section we relate string classes and 
leading order Chern classes to genus zero Gromov-Witten invariants 
and gravitational descendants associated 
to characteristic classes. We also investigate when the integrals that often
denote GW invariants are rigorous expressions. In particular, we want to realize GW invariants as integrals of forms over the moduli space ${\mathcal M}_{0,k}(A)$ defined below, without using the compactification
$\overline{{\mathcal M}}_{0,k}(A).$  Thus we want to avoid both the construction of the virtual fundamental class and 
discussions of non-smooth  Poincar\'e duality spaces. This is certainly not possible in general, so we restrict ourselves mainly to the 
semipositive case, where the moduli space of pseudoholomorphic curves has an especially
nice compactification.

Recall that GW invariants are built from cohomology classes 
$\alpha_i$ on the target manifold $M$, while gravitational descendants \cite[Ch. 10]{CK}
also involve $\psi$ classes, which are first Chern classes of 
line bundles on the moduli space of marked curves.  We first relate gravitational descendants to string classes (Thm. \ref{gd}), 
and then relate GW invariants built from even classes on the target manifold to string and leading order classes (Thm. \ref{pp}).  For most of this section, we work in the symplectic setting.  At the end, we make some comments about the algebraic case and the role of the virtual fundamental class.

 \subsection{Notation}

Following the notation in \cite{ms}, let $M$ be a closed symplectic manifold with a generic compatible almost complex structure.
For $A\in H_2(M,\Z)$, set
$C^\infty_0(A) = \{f:\BbP^1\to M| f\in C^\infty, f_*[\BbP^1] = A\}.$  Let $G= \Aut(\BbP^1)\simeq 
PSL(2,\C)$ be the group of complex automorphisms of $\BbP^1$.
Set $\BbP_k^1 = \{(x_1,\ldots, x_k)\in (\BbP^1)^k: x_i\neq x_j\ {\rm for}\ i\neq j\}.$
For
fixed $k \in \Z_{\geq 0}$, set
 $$C^\infty_{0,k}(A) = (C^\infty_0(A) \times \BbP_k^1)/G,$$
  where the action of $G$ on $C^\infty_g(A) $ is 
  $\phi\cdot f = f\circ \phi^{-1}$ and $\phi$ acts diagonally on $\BbP_k^1.$
Denoting an element of $C^\infty_{0,k}(A)$ by $[f, x_1,\ldots, x_k]$, we set the moduli space 
of pseudoholomorphic maps to be 
$\calM_{0,k}(A) = \{[f, x_1,\ldots, x_k]: f\ {\rm is\ pseudoholomorphic}\}.$  $\mgk$ is
a smooth, finite dimensional, noncompact manifold.

We impose the condition that all maps $f$ are simple, {\it i.e.} $f$ does not factor through a branched covering map from $\BbP^1$ to $\BbP^1.$  In this case, the action of
$G$ on $C^\infty_0(A) \times \BbP_k^1$ is free, and 
$C^\infty_{0,k}(A)$ is an infinite dimensional manifold of either Banach or Fr\'echet type.  


The forgetful map $\pi = \pi_k:C^\infty_{0,k}(A)\to C^\infty_{0,k-1}(A)$  given by
$[f, x_1,\ldots, x_{k-1}, x_k]\mapsto [f, x_1,\ldots, x_{k-1}]$   is a locally trivial smooth 
fibration, since 
for  disjoint neighborhoods $U_1,\ldots, U_{k-1}$ around a fixed $x_1,\ldots, x_{k-1}$, we have
$$\pi^{-1}[ (C^\infty_0(A)\times \prod U_i)/G]\approx 
[ (C^\infty_0(A)\times \prod U_i)/G]\times
\BbP^1_-,$$
{\it i.e.}, the fiber $\BbP^1_- = \BbP^1\setminus \{x_1,\ldots, x_{k-1}\}$ consists of all choices for the $k^{\rm th}$ point.
This fibration restricts to a fibration on the moduli spaces, but  
 does not extend to compactifications of the moduli spaces.

Let $L_i$ be the line bundle over $C^\infty_{0,k}(A)$ with fiber $T^*_{x_i}\BbP^1$ over
$[f, x_1, \ldots, x_k].$  This bundle is well defined, since an automorphism $\phi$ gives an identification of tangent spaces $d\phi^*_{x_i}:T_{\phi(x_i)}^*\BbP^1\to T^*_{x_i}\BbP^1.$ 
We set $\calL_i = \pi_*L_i$ for $i = 1,\ldots, k.$  The fibers of $\calL_i$ are given by
\begin{eqnarray*} \calL_i|_{[f, x_1,\ldots,x_{k-1}]} &=& \{s:\BbP_-^1\to T_{x_i}^*\BbP^1\},\ i =1,\ldots,k-1,\\
\calL_k|_{[f, x_1,\ldots,x_{k-1}]} &=& \Gamma(T^*\BbP_-).
\end{eqnarray*}

If we put a Sobolev topology on $\cok$ ({\it i.e.} we consider two maps close if their first $s$ partial derivatives are close for a fixed $s \gg 0$), then $\cok$ is a Banach manifold and so admits partitions of unity.  Thus the $L_i$ have connections.  
In any case, we are interested in $L_i$ restricted to the finite dimensional manifold
$\calM_{0,k}(A)$, so the existence of connections is not an issue.  

We want to choose connections that do not blow up as $x_i\to x_j$ in $\BbP_k^1.$
 Since the line bundle $L'_i$ with fiber $T^*_{x_i}\BbP^1$ is a well defined line bundle on $(\BbP^1)^k/G$, it 
restricts to a line bundle on $\BbP_k^1$, which then pulls back to the bundle $L_i$ on $\cok.$  We always take connections on $L'_i$, as these restrict to connections 
on $L_i$ which are well behaved on $\BbP_k^1.$  

The $\calL_i$ have leading order first Chern classes, but it is better to consider 
the associated string classes.  Definition \ref{2.2} in the current context is as follows.

\begin{defn}  {\it The string class $c_1^{{\rm str},r}(\calL_i)\in H^{2r-2}(\calM_{0,k-1}(A))$ or 
in $H^{2r-2}(\cok)$
is the de Rham class of }
$$ \int_{\BbP'} [\Tr(\Omega_i)]^{r}, $$
{\it where $\Omega_i$ is the curvature of a restricted connection on $L_i.$ Thus
$c_1^{{\rm str},r}(\calL_i) = \pi_* (c_1(L_i)^{r})$, where $\pi_*$ is the pushforward map given by integration over the fiber.}
\end{defn}

Since $[\Tr(\Omega_i)]^r$ is a closed form, the right hand side of the definition is 
closed.   Here we use the fact that for restricted connections on $L_i$,
the integral over the fiber exists and extends to the compact space $\BbP^1.$  The usual arguments that $c_1^{{\rm str},r}$ is closed (which uses Stokes' Theorem on 
$\BbP^1$)
with de Rham class independent of the connection carry over.  

Note that for $r=1$, the string class $c_1^{{\rm str}}(\calL_i)\in H^0$ is the
(constant function) $-2+k$, since $\int_{\BbP'} \Tr(\Omega)$ equals  
$\chi(T^*\BbP') = -\chi(T\BbP').$

Let $ \ev^k:\cok\to M^k$ be $\ev^k[f, x_1, \ldots, x_k ] = (f(x_1), \ldots, f(x_k))$, let $p_i:M^k\to M$ 
 be the projection onto the $i^{\rm th}$ factor, and set $\ev_i^k = p_i\circ \ev^k:\cgk\to M.$  
Then $\ev_i^k = \ev_i^{k-1}\circ \pi$ for $i<k.$ When the context is clear, we will denote $\ev^k$ by $\ev.$

\subsection{Semipositive manifolds}

We will give cases when GW invariants and gravitational descendants can be detected by integration of forms over the moduli space $\mok$,  without using  the compactification $\bmok.$  
This is expected to happen if the boundary strata have codimension at least two in $\bmok$, {\it e.g.}, if
$M$ is semipositive  \cite[\S6.4]{ms3}.  The main results of this section justify this integration over just 
$\mok.$

{\it For the rest of this section, except for the remarks at the end, we assume that $M$ is 
semipositive.}  For motivation, we first pretend that $\bmok$ carries a fundamental
class.  Then for $\alpha_i\in H^{d_i}(M,\C)$ of appropriate degree, 
the GW invariant associated to the $\alpha_i$ is
\begin{eqnarray}\label{quick} \langle \alpha_1\ldots\alpha_k\rangle & \stackrel{\rm def}{=}& \ev_*[\bmok]\cdot a
= \langle \PD \ev_*[\bmok]\cup \alpha,[M^k]\rangle\\
&=& \langle \alpha,\ev_*[\bmok]\rangle = \langle \ev^*\alpha, [\bmok]\rangle.\nonumber
\end{eqnarray}
Here $\alpha = \alpha_1\times\ldots\times\alpha_k \in H^*(M_k)$ and 
$a = \PD\  \alpha$ is the Poincar\'e dual of $\alpha.$  In (\ref{quick}), we use the characterization of 
Poincar\'e duality: 
for $a\in H_*(X,\Z), \beta\in H^*(X,\Z)$ on an oriented compact manifold $X$, 
\begin{equation}\label{PDPD}\langle \alpha \cup \beta, [X]\rangle = \langle \beta, a\rangle.
\end{equation}

To do this more precisely, we follow the careful exposition in \cite{zi}.  
For $M$ semipositive, 
$$\partial \ev= \bigcap\limits_{K\ {\rm compact}} \overline{\ev(\mok\setminus K)}\ \ \subset M^k$$
 lies in the image of a map of a manifold of dimension at most 
$ \dim \mok - 2 := 
r-2$.  Recall that $r = \dim M + 2c_1(A) + 2k-6,$   with $c_1(A) = \langle c_1(M),A\rangle.$
Thus by definition
$\mok$ 
defines 
a pseudocycle in $M^k$.
By \cite[Prop 2.2]{zi}, there exists an open set $U = U_k$ in $M^k$ with
\begin{equation}\label{uk}
\partial\ev \subset U\subset M^k,\ \ 
H_r(M^k,U; \Z) \simeq H_r(M^k;\Z)
\end{equation}
(see (\ref{uk2})).
  Let $\bar V$ be a compact manifold with boundary inside $\mok$ with 
  $$\mok\setminus\ev^{-1}(U) \subset \bar V;$$
    we think of $\bar V$  as 
``most of" $\mok.$  
Specifically,  $\ev_*[\bar V] \in H_r(M^k,U; \Z) \simeq H_r(M^k;\Z)$ is a substitute for the ill-defined $\ev_*[\bmok]$.

\begin{defn}  {\it
The GW invariant associated to $\alpha$ is  
 $$\langle \alpha_1\ldots\alpha_k\rangle = \ev_*[\bar V]\cdot a\in \Z,$$
  provided $\sum_i |\alpha_i| = 
 k\dim M -r,$ for $|\alpha_i|$ the degree of $\alpha_i.$  }
 \end{defn}

 The GW invariant is independent of the choice of $U$ and $V$. More generally, we can take positive integers $\ell_i$ with
 $\sum_i \ell_i  |\alpha_i| = 
 k\dim M -r$, take $a = 
 \PD\ \alpha) = \PD (\alpha_1^{\ell_1}\times\ldots\times \alpha_k^{\ell_k})$, and similarly define $\langle\alpha_1^{\ell_1}\ldots \alpha_k^{\ell_k}\rangle$.
 
There is an integer $q$ such that 
$qa$ has a representative 
submanifold $N$; if $N$ is unorientable, we have to pass its oriented double cover.  Of course, $\PD(N) = \alpha$, but we can represent $\alpha$
by a compactly supported closed form, the Thom class of the normal bundle of $N$ in $M^k$, thought of as a tubular neighborhood of $N$.  
Then
$$\langle\alpha_1^{\ell_1}\ldots \alpha_k^{\ell_k}\rangle = \frac{1}{q} \ev_*[\bar V]\cdot N = \frac{1}{q}\langle \PD( N), \ev_*[\bar V]\rangle =
\frac{1}{q} \langle \ev^* \PD(N), [\bar V]
\rangle.
$$
In the last term, $[\bar V] \in H_r(\mok, \mok-V; \Z), \ev^*\PD(N)\in H^r(\mok,\R).$  Since $\ev^*\PD(N)$ is a differential form, we can write 
\begin{equation}\label{en}\langle\alpha_1^{\ell_1}\ldots \alpha_k^{\ell_k}\rangle = \frac{1}{q}\int_{\mok}\ev^*\PD(N)
\hskip 0.3 in  ({\rm mod}\  \mok\setminus V),\end{equation}
where mod $\mok\setminus V$ means $\int_K\theta = 0$ for a form $\theta$ and a submanifold, possibly with boundary,
$K\subset \mok\setminus V.$  The modding out ensures that (\ref{en}) is independent of the representative of $\PD(N).$  This justifies writing a GW invariant as the
 integral of a form over
$\mok.$

From now on, we assume $q=1$ for convenience and drop 
``mod $\mok\setminus V$".

To bring in the bundle $\calL_i$, define
the gravitational descendant(or gravitational correlator) 
associated to classes $\alpha_1\ldots \alpha_k\in H^*(M, \C)$ and 
 multi-indices $(\ell_1,\ldots, \ell_k), (r_1,\ldots, r_k)$
by
$$\langle t_1^{r_1}\alpha_1^{\ell_1}\ldots
 t_k^{r_k}\alpha_k^{\ell_k}\rangle _{0,k}
= \int_{\mok} c_1(L_1)^{r_1}\wedge\alpha_1^{\ell_1}\wedge\ldots\wedge 
c_1(L_k)^{r_k}\wedge\alpha_k^{\ell_k}.
$$
Here  $\sum_i 2r_i + \ell_i|\alpha_i| = {\rm dim}\ \mok.$  

We set
\begin{eqnarray*}\lefteqn{\langle t_1^{r_1}\alpha_1^{\ell_1}
 \ldots t_{k-1}^{r_{k-1}}c_1^{{\rm str}, r_k}\rangle_{0,k-1} }\\
 &=& \int_{\mokm} 
 c_1(L_1)^{r_1}\wedge\ev_1^*\alpha_1^{\ell_1}\wedge\ldots\wedge 
c_{1}(L_{k-1})^{r_{k-1}}\wedge\ev_{k-1}^*\alpha_{k-1}^{\ell_{k-1}}\wedge
c_1^{{\rm str},r_k}(L_k).
\end{eqnarray*}

\begin{thm} \label{gd} For $\ell_k=0$, the gravitational descendents satisfy
$$\langle t_1^{r_1}\alpha_1^{\ell_1}
 \ldots \alpha_{k-1}^{\ell_{k-1}}t_k^{r_k}\rangle_{0,k} = 
 \langle t_1^{r_1}\alpha_1^{\ell_1}
 \ldots t_{k-1}^{r_{k-1}}\alpha_{k-1}^{\ell_{k-1}}c_1^{{\rm str},r_k}
 \rangle_{0,k-1}.$$
 \end{thm}
  
\begin{cor} (Dilaton Axiom) \cite[p. 306]{CK}  
{For $r_k=1,$ }  
$$\langle t_1^{r_1}\alpha_1^{\ell_1}
 \ldots \alpha_{k-1}^{\ell_{k-1}}t_k\rangle_{0,k} = (-2+k)
 \langle t_1^{r_1}\alpha_1^{\ell_1}
 \ldots t_{k-1}^{r_{k-1}}\alpha_{k-1}^{r_{k-1}}\rangle_{0,k-1}.$$
\end{cor}

\noindent {\it Proof of the Theorem.}  For a fibration $Z\to M\stackrel{\pi}{\to} B$ of smooth compact
manifolds, we have for $\omega\in \Lambda^*(B), \eta\in \Lambda^*(M)$,
\begin{equation}\label{PD2}\int_M \pi^*\omega\wedge \eta = \int_B\omega\wedge \pi_*\eta.
\end{equation}
In particular, this holds for $\BbP^1_-\to\mok\stackrel{ \pi_k}{\to}\mokm$ provided the forms extend to 
the closures of the moduli spaces.

For this fibration, we can  extend (\ref{PD2}) to integration mod $\mok\setminus V$ provided
\begin{equation}  \label{pro}
\pi_k(\mok\setminus V_k)\subset \mokm\setminus V_{k-1},
\end{equation}
where appropriate subscripts for $V$ have been added.  

We first sketch the proof of  (\ref{pro}).   By \cite[Lemma 2.4]{zi}, for a fixed 
generic triangulation $T$ of 
$M^k$ with simplices $\sigma$, the set $U=U_k$ in (\ref{uk}) is given by
\begin{equation}\label{uk2}
U_k  = \bigcup_{m= \dim M^k - \dim\mok}^{\dim X^k}\ \ \  \bigcup_{\sigma, |\sigma|=m} {\rm St}(b_\sigma, {\rm sd}\  T),
\end{equation}
where $b_\sigma$ is the barycenter of $\sigma$, sd is the first subdivision of $T$, and 
the star ${\rm St}(b_\sigma, {\rm sd}\  T)$ consists of the interior of all simplices in
sd $K$ containing $b_\sigma.$ (We don't distinguish between the simplices in the 
triangulation and their images in $M^k.$)
By the proof of \cite[ Lemma 2.4]{zi}, we can restrict the simplices in $U_k$ by any subset of 
$\{\sigma: |\sigma| \geq \dim M^k - \dim\mok\}$, provided we keep all the top simplices
that meet $\partial\ev_k.$  Given $U_{k-1}$, we will suitably restrict $U_k$ so that 
$p(U_k)\subset U_{k-1}$, for $p:M^k\to M^{k-1}$ the projection onto the first $k-1$ factors.  Since $\pi_k$ is a fibration, we will conclude that 
$x\not\in V_k$ implies $\pi_k(x)\not\in V_{k-1}$, which  is (\ref{pro}).

To fill in the details of (\ref{pro}),  we note that
\begin{equation}\label{cd2}\begin{CD} \mok@>\ev^k>>M^k\\
@V\pi_kVV@VVp V\\
\mokm@>>\ev^{k-1}> M^{k-1}
\end{CD}
\end{equation}
commutes, with  $\pi_k$ and $p$ surjective and open. 
 It follows that for $X\subset M^{k-1}$,
\begin{equation}\label{set}
(\ev^{k-1})^{-1}(X) = \pi_k( \ev^k)^{-1}p^{-1}(X).
\end{equation}

We claim that
$$p(\partial\ev^k)\subset \partial\ev^{k-1}.$$
For the claim, recall that 
$\partial\ev^{k-1} = \ev^{k-1}(\partial\mokm)$, and
points in $\partial\mokm$ are  given by $Z'=[h, y_1,\ldots, y_{k-1}]$, for $h$ a pseudoholomorphic maps on smooth curves or cusp curves ({\it i.e.,} curves with bubbling), with the $y_j$ possibly coincident.  Choose such a $Z'$ with $\ev_{k-1}(Z') = Y'.$  Take
compact sets $K_i'$ exhausting $M^{k-1}$, and choose 
$F^i = \ev^{k-1}[f^i, x_1^i,\ldots, x_{k-1}^i]\not\in K_i'$ with $\lim_i F^i = Y'.$
Take $G^i
= \ev^k[f^i, x_1^i,\ldots, x_k^i]
\in \pi^{-1}(F_i)$ for  $x^i_k$ distinct from the other $x_j^i$ and with 
$\lim_i f^i(x_k) = \lim_if^i(x_{k-1})$; note that the last limit exists.  Then 
$\lim_i G^i = Z\in \partial\mok$ exists, and for $Y = \ev^k(Z)\in \partial\ev^k$
$$p(Y) = p\ev^k(Z) = \ev^{k-1}\pi_k(Z) =\ev^{k-1} Z' = Y'.$$



As mentioned above, choose $U_{k-1}$ to contain only those top simplices $\sigma^{\rm top}$ which
meet $\partial\ev^{k-1}.$  Refine the triangulation of $M^k$ to a new triangulation, also called $T$, so that each
$p^{-1}(\sigma^{\rm top}_{k-1})$ is the sum of top simplices in $T$.  Set $U_k$ to contain 
only the simplices in $p^{-1}(U_{k-1})$; by the claim above, $U_k = p^{-1}(U_{k-1})$ is an open neighborhood of $\partial\ev_k,$   and
$$x\in U_k\Leftrightarrow p(x)\in U_{k-1}.$$

By (\ref{set}), 
$\pi_k\ev_{k}^{-1}(U_{k}) = \ev_{k-1}^{-1}(U_{k-1}).$  Thus for a choice of compact 
manifold with boundary
$\bar V_{k-1}\subset \mokm\setminus \ev_{k-1}^{-1}(U_{k-1})$ and a
slight perturbation of $\pi_k,$ $\bar V_k=  \pi^{-1}_k(\bar V_{k-1})$ is a
manifold with boundary containing $\mok\setminus \ev_{k-1}^{-1}(U_{k-1})$.  Finally,
$\bar V_k$ misses a smaller open neighborhood of $\bmok$, so the closed set $\bar V_k$ is contained in a compact subset of $\mok$.  Thus $\bar V_k$ is compact.  

By this construction, (\ref{pro}) is satisfied, so we can apply (\ref{PD2}) to the fibration
$\BbP^1_-\to\mok\stackrel{ \pi_k}{\to}\mokm$.

By (\ref{cd2}), $\pi_k^*\ev_{k-1}^* = \ev_k^*\pi^*$, which we abbreviate by dropping 
$\pi$ 
and denoting $\pi_k$ by $\pi$:   $\pi^*\ev_{k-1}^* = \ev_k^*.$
Since $L_{i,k}$, that is, $L_i$ as a bundle over $\mok$, satisfies
$\pi^*L_{i,k-1} = L_{i,k}$  
for $i<k$,  
by (\ref{PD2}) we have (with even more subscripts omitted)
\begin{eqnarray*} 
\lefteqn{ \langle t_1^{r_1}\alpha_1^{\ell_1}
 \ldots \alpha_{k-1}^{\ell_{k-1}}t_k^{r_k}\rangle_{0,k} }\\
 &=& \int_{\mok} \pi^*\left(c_1(L_1)^{r_1}\wedge \ev^*_1\alpha_1^{\ell_1}\wedge
 \ldots
 \wedge c_1(L_{k-1})^{r_{k-1}}\wedge \ev^*_{k-1}\alpha_{k-1}^{\ell_{k-1}}
\right) \wedge c_1(L_k)^{r_k}\\
&=& \int_{\mokm} 
 c_1(L_1)^{r_1}\wedge \ev^*_1\alpha_1^{\ell_1}\wedge\ldots
 \wedge c_1(L_{k-1})^{r_{k-1}}\wedge \ev^*_{k-1}\alpha_{k-1}^{\ell_{k-1}}
\wedge \pi_* c_1(L_k)^{r_k}\\
&=&  \langle t_1^{r_1}\alpha_1
 \ldots t_{k-1}^{r_{k-1}}\alpha_{k-1}^{\ell_{k-1}}c_1^{{\rm str},r_k}\rangle_{0,k-1}.
 \end{eqnarray*}
${}$\hfill $\Box$
\bigskip

For pure GW invariants, we investigate the geometry of the fibration $\pi.$  The next series
of lemmas hold for moduli spaces of genus $g$ curves.

\begin{lem}\label{fiveone}  $\pi$ is flat; i.e., for each $[f, x_1,\ldots, x_k]\in C^\infty_{0,k}(A)$,
there exists an integrable distribution 
$H_{[f,x_1,\ldots, x_k]}\subset T_{[f,x_1,\ldots, x_k]}C^\infty_{0,k}(A)$ such that
$$T_{[f,x_1,\ldots, x_k]}C^\infty_{0,k}(A)\simeq {\rm ker}(\pi_*)
\oplus H_{[f,x_1,\ldots, x_k]}.$$
\end{lem}

\begin{cor}  $\pi:\calM_{0,k}(A)\to \calM_{0,k-1}(A)$ is flat.
\end{cor}

\noindent {\it Proof of the Lemma.}  Let $\gamma_0(t)$ be a curve in $C^\infty_{0,k}(A)$ with $\gamma_0(0) = f$, and 
let $\gamma(t)$ be curves in $\BbP^1$ with $\gamma_i(0) = x_i$ for $i = 1,\ldots, k-1.$
  Set
\begin{eqnarray*}
H_{[f,x_1,\ldots, x_k]} &=& \left\{ (d/dt|_{t=0}) [\gamma_0(t), \gamma_1(t),\ldots,
\gamma_{k-1}(t), x_k]: \gamma_0,\ldots, \gamma_{k-1}\ {\rm as\ above}\right\}\\
&=& \left\{  [\dot \gamma_0(0), \dot\gamma_1(0),\ldots,
\dot \gamma_{k-1}(0), 0]\right\}.
\end{eqnarray*}
Since 
$$(d/dt|_{t=0}) [\gamma_0(t), \gamma_1(t),\ldots,
\gamma_{k-1}(t), x_k] = (d/dt|_{t=0}) [\gamma_0(t)\circ \phi^{-1}, \phi\circ \gamma_1(t),\ldots,\phi\circ \gamma_{k-1}(t), x_k],$$
$H$ is well defined. 

Let $X^h$ be the horizontal lift of a vector field $X$ on  $C^\infty_{0,k-1}(A)$, so 
$X^h$ is of the form\\ $[X_0, X_1,\ldots, X_{k-1},0].$
The Lie bracket
$\calL_{X^h}Y^h$ is computed using the flow of $X^h$, which is locally of the form 
$[\eta_0(t),\eta_1(t),\ldots, \eta_{k-1}(t), x_k].$  Thus $X=
[\dot\eta_0(0),\dot\eta_1(0),\ldots,\dot \eta_{k-1}(0)]$ and similarly for $Y$, so
\begin{equation}\label{h1}[X^h, Y^h] = \calL_{X^h}Y^h = [\calL_XY, 0] = [X,Y]^h\in H.
\end{equation}
(The second bracket in (\ref{h1}) refers to a point in $\mok.$)
${}$\hfill$\Box$
\medskip

\begin{rem}  In the local proofs of the families index theorem \cite{BGV}, the setup is a
fibration $Z\to M\stackrel{\pi}{\to} B$ of smooth manifolds with a horizontal distribution, and a hermitian bundle with connection
$(F,\nabla)\to M.$  To this data one can associate the infinite rank bundle $\calF = \pi_*F$ with the pushforward connection $\nabla' = \pi_*\nabla$. The curvature of $\nabla'$ in general
takes values in first order differential operators acting on the fibers 
$\Gamma(F|_{\pi^{-1}(b)})$ of $\calF$.
It is difficult to construct Chern classes for these bundles, as there are no known nontrivial traces on first order operators.  Thus it is particularly significant to have a flat 
fibration.  In this case, we will show that the curvature takes values in zeroth order operators, and so
has a leading order trace.
\end{rem}
\medskip

Let $(F,\nabla)$ be a finite rank hermitian bundle with connection over $\cok$, so 
$\calF = \pi_*F$ is an infinite rank bundle over $\cokm.$  $\pi_*F$ has the connection 
$\nabla'$ defined on $s\in \Gamma(\pi_*F)$ by 
$$ \nabla'_X s [f, x_1,\ldots, x_{k-1}] = \nabla_{X^h} \tilde s [f, x_1, \ldots, x_k],
$$
where $\tilde s [f, x_1, \ldots, x_k] = s [f, x_1,\ldots, x_{k-1}, x_k].$
By (\ref{h1}), the curvature of $\nabla'$ is given by
\begin{eqnarray*}\Omega'(X,Y) &=& \nabla'_{X^h}\nabla'_{Y^h} -
\nabla'_{Y^h}\nabla'_{X^h} -\nabla'_{[X, Y]^h}\\
&=& \nabla'_{X^h}\nabla'_{Y^h} -
\nabla'_{Y^h}\nabla'_{X^h} -\nabla'_{[X^h, Y^h]}\\
&=&\Omega(X^h, Y^h),
\end{eqnarray*}
where $\Omega$ is the curvature of $\nabla.$  In summary, we have
\begin{lem} \label{fivethree} Let $(F,\nabla)\to\cgk$ be a bundle with connection with curvature $\Omega.$
In the notation of Lemma \ref{fiveone}, the induced connection $\nabla'$ on
$\calF = \pi_*F$ has curvature $\Omega'(X,Y) = \Omega(X^h, Y^h)$.
\end{lem}

Let $\alpha_i$ be elements of the {\it even} cohomology of $M$.  Since the Chern character $ch:K(M)\otimes \C\to H^{\rm ev}(M,\C)$ is an isomorphism, $\alpha_i
= ch(E_i)$ 
for a virtual bundle $E_i$ ($i = 1,\ldots, k-1$), and $\alpha_k^{\ell_k} = ch(E_k)$.
Pullbacks and pushdowns of the $E_i$ are well defined 
virtual bundles.

\begin{thm} \label{pp} Let $\alpha_i\in H^{\rm ev}(M,\C)$ satisfy $\alpha_i = ch(E_i)$, $i=1,\ldots k-1$,
 and let $\alpha_k^{\ell_k} = ch(E_k)$ for $E_i\in K(M)$.  Set
$\calE_i = \pi_* \ev_i^*E_i\to\mokm.$  Then
\begin{eqnarray*}
\langle\alpha_1^{\ell_1}\ldots\alpha_{k}^{\ell_{k}}\rangle_{0,k}  
&=&  \langle ch^{\rm lo}(\calE_1)\cdots
ch^{\rm lo}(\calE_{k-1})ch^{\rm str}(\calE_k)\rangle_{0,k-1}.
\end{eqnarray*}
\end{thm}

\begin{proof} 
Pick connections $\nabla_i$ on $E_i$ with curvature $\Omega_i$. Then
\begin{eqnarray}\label{i} 
\lefteqn{ \langle\alpha_1^{\ell_1}\ldots\alpha_k^{\ell_k}\rangle_{0,k}}\\
&=& \int_{\mok} \ev_1^*([\Tr(\exp(\Omega_1)]^{\ell_1})\wedge\ldots\wedge 
\ev_{k-1}^*(
[\Tr(\exp\Omega_{k-1})]^{\ell_{k-1}}\wedge \ev_{k}^*
([\Tr(\exp\Omega_{k})])\nonumber\\
&=& \int_{\mok} [\Tr(\pi^*\ev_1^*\exp(\Omega_1))]^{\ell_1}\wedge\ldots\wedge
[\Tr(\pi^*\ev_{k-1}^*\exp(\Omega_{k-1}))]^{\ell_{k-1}}\wedge [\Tr(\ev_k^*\exp
\Omega_k)]
\nonumber\\
&=& \int_{\mokm}  [\Tr(\ev_1^*\exp(\Omega_1))]^{\ell_1}\wedge\ldots\wedge
[\Tr(\ev_{k-1}^*\exp(\Omega_{k-1}))]^{\ell_{k-1}}\wedge \pi_*[\Tr(\ev_k^*
\exp(\Omega_k))]. \nonumber
\end{eqnarray}

We have 
$$\pi_*\left([\Tr(\ev_k^*\exp(\Omega_k))]^{\ell_k}\right) = ch^{\rm str}(\calE_k),$$
where $\calE_k = \pi_*\ev_k^*E_k.$

We claim that in the last line of (\ref{i}), 
$ [\Tr(\ev_i^*\exp(\Omega_i))]$ is the leading order Chern 
character
 for $\calE_i = \pi_*E_i$, for $i<k.$  Dropping the index $i$, let 
 $\ev_u:\mok\to M, \ev_d:\mokm\to M$ be the $i$-th evaluation maps.  
 Then the leading order Chern character
 for
$\calE$ as a differential form is given by
$$ch^{\rm lo}(\calE)(X_1,\ldots, X_{2r}) = \int_{\BbP'} \Tr(\exp(\Omega^{\ev^*_uE}))
(X_1^h,\ldots, X^h_{2r})\dvol_{\BbP'},$$
by (\ref{trlo}).   Since
\begin{eqnarray*}\Tr(\exp(\Omega^{\ev^*_uE}))_z(X_1^h,\ldots, X^h_{2r})
&=& \Tr(\exp(\Omega^{\pi^*\ev^*_dE}))_z(X_1^h,\ldots, X^h_{2r})\\
&=& \Tr(\exp(\Omega^{\ev^*_uE}))_z(\pi_*X_1^h,\ldots, \pi_*X^h_{2r})\\
&=& \Tr(\exp(\Omega^{\ev^*_uE}))(X_1,\ldots, X_{2r})
\end{eqnarray*}
is independent of  $z\in\BbP'$, we get
$$ch^{\rm lo}(\calE)= \int_{\BbP'} \Tr(\exp(\Omega^{\ev^*_uE}))
\dvol_{\BbP'} = {\rm Vol}(\BbP')\Tr(\exp(\Omega^{\ev^*_uE}))
= {\rm Vol}(\BbP')\Tr(\ev^*\exp(\Omega^{E})).$$
Setting the volume of $\BbP'$ equal to one finishes the claim and the propf.
\end{proof}
\bigskip

We briefly discuss the algebraic setting.  $M$ is now a smooth projective variety, with 
$\mgk,$
$ \bmgk$  the moduli space/stack of stable maps of a fixed genus $g$ curve into $X$ representing $A\in H_2(M)$, and its compactification, respectively.  The
forgetful map 
$\pi:\mgk\to\mgkm$ exists as long as $n+2g\geq 4$ or $A\neq 0$ and $n\geq 1$
\cite[p. 183]{CK}.
Provided the open moduli spaces are oriented manifolds, we can represent a suitable multiple of $a\in H_*(\mgkm,\Z)$ by an submanifold $N$, and then
$$\int_N \omega\wedge\pi_*\eta = \int_{\pi^{-1}N}\pi^*\omega\wedge\eta,$$
for $\omega\in \Lambda^*(\mgkm), \eta\in \Lambda^*(\mgk)$  compactly supported forms.

For a fibration $M\stackrel{\pi}{\to} B$ of oriented compact manifolds, define the homology pullback \\
$\pi^*:H_*(B)\to H_*(M)$ by $\pi^* = \PD_M^{-1}\circ \pi^*\circ \PD_B$, where $\pi^*$ on the right hand side is the usual cohomology pullback. 
By (\ref{PDPD}), (\ref{PD2}), 
\begin{equation}\label{PD3}\int_N\omega\wedge\pi_*\eta = \int_B\omega\wedge\pi_*\eta\wedge \PD_B(N)
=\int_M \pi^*\omega\wedge\eta\wedge\pi^*\PD_B(N) 
= \int _{\pi^*N}\pi^*\omega\wedge\eta,
\end{equation}
where we identify $N$ with its homology class.

We would like to apply (\ref{PD3}) to $[N] = \bvm,$ the virtual fundamental class of
$\overline{\calM}_{g,k-1}(A).$ Provided $\overline{\calM}_{g,k-1}(A), \overline{\calM}_{g,k}(A)$ are orbifolds, we have
$$\int_{\bvm} \omega \wedge\pi_*\eta = \int_{\bv}\pi^*\omega \wedge\eta,$$
since
$\pi^*\bvm = \bv$ \cite[(7.22)]{CK}.  The moduli spaces are orbifolds if e.g. 
$g=0$ and $M = \BbP^n$  (or more generally if $M$ is convex) \cite{FP}. 
In these cases, Theorem 4.1 continues to hold, since string classes are given by topological pushforwards.  However, for Theorem 4.6,
we would need to know in addition that $\ev_i^*E_i$ admit (suitable variants of) connections over the compactified moduli spaces, since the leading order classes are constructed from connections.  It is not clear that this is possible.

\section{Leading order Chern-Simons classes on loop groups}\label{string}

In this section we return to loop spaces in  the special case of based loop groups $\Omega G$.  We define 
Chern-Simons analogues of the string and leading order characteristic forms of \S2.  
The main result is that  $H^*(\Omega G, \R)$ is generated by Chern-Simons string classes (Thm.~\ref{thmfive}) or by a pushforward of
of leading order Chern-Simons classes (Thm.~\ref{relate}).  The relation between these different approaches is also given in 
Thm.~\ref{relate}.

\subsection{Relative Chern-Simons forms on compact Lie groups}
We first express the generators of the cohomology ring of a compact Lie group in terms of relative Chern-Simons classes.

Let $E\to M$ be a rank $n$ vector bundle over a closed manifold  $M$, and let $\nabla_0,\nabla_1$ be connections on $E$. Locally, we have $\nabla_i = d+\omega_i$,
and $\omega_1-\omega_0$ is globally defined.
Let $\Omega_t$ be the curvature of the  connection $\nabla_t=t\nabla_0+(1-t)\nabla_1$.
For a subgroup $G'$ in $U(n)$ and an Ad${}_{G'}$-invariant
polynomial $f$, 
the relative Chern-Simons form 
\begin{equation}\label{relcs}
CS_f(\nabla_0,\nabla_1)=\int_0^1f(\omega_1-\omega_0,\Omega_t,\ldots,\Omega_t)dt
\end{equation}
satisfies 
$f(\Omega_1)-f(\Omega_0)=dCS_f(\nabla_0,\nabla_1)$, provided the $\nabla_i$ are
$G'$-connections. 

Assume now that $E$ is a trivialized bundle and $\nabla_0=d$ is the canonical flat connection. A gauge transformation $h\in{\rm Aut}(E)$ yields the flat connection
 $\nabla_1=h\cdot\nabla_0=h^{-1}d h$. 
 We have $\omega_1-\omega_0=h^{-1}dh$ and $\Omega_t=C_t[h^{-1}dh,h^{-1}dh]$ for  $C_t > 0$. Note the confusing notation for $\nabla_1=h^{-1} dh$ and the 
 global connection one-form $h^{-1}dh.$
 In any case, $CS_f$ is a closed form of odd degree 
 and so determines a Chern-Simons cohomology class.

Let $G$ be a compact Lie group with Lie algebra $\mathfrak{g}$ and Maurer-Cartan form
$\theta^G$.
Choose a   finite dimensional faithful unitary representation $h:G\to\textrm{Aut}(V)$
with ${\rm Im}(h) = G'$. We may assume that $h$ is the exponentiated version of a faithful Lie algebra representation $dh = h_*: \mathfrak{g}\to\End(V).$  For $\underline{V}\to G$ the trivial vector bundle $G\times V\to G$, we can view 
$h$ as a gauge transformation of $\underline{V}.$
Let $\nabla_0=d$ and $\nabla_1=h^{-1}d h$ be connections on $\underline{V}$ as above. 
 As before,  $\omega_1-\omega_0=h^{-1}dh$, where 
$$h^{-1}dh|_g(g_*X)=h^{-1}(g)\,dh_g(g_*X)=h^{-1}(g)
 h(g)h_*(X)=h_*(X),$$
for  $X\in\mathfrak{g}$.  Here $g_* = (L_g)_*$ is the differential of left multiplication by $g$.  
Thus $h^{-1}dh=h_*(\theta^G)\in \Lambda^1(G, {\rm End}(\underline{V}))$.

An  Ad${}_G$-invariant polynomial $f$ on $\mathfrak{g}$ determines an $h(G)$-invariant polynomial $h_*f$ on $ \textrm{Im}(h_*)\subset \textrm{End}(V)$
by
$$(h_*f)(\alpha_1,\alpha_2,\ldots,\alpha_k) = f(h_*^{-1} \alpha_1,\ldots,h_*^{-1} \alpha_k).$$ 
In particular, 
 $$f(\theta^G,[\theta^G, \theta^G],\ldots,[\theta^G, \theta^G])= (h_*f)(h^{-1}dh,[h^{-1}dh, h^{-1}dh],\ldots,[h^{-1}dh, h^{-1}dh]).$$
The corresponding Chern-Simons form is equal to
$$CS_f(\nabla_0,\nabla_1)= C(h_*f)(h^{-1}dh,[h^{-1}dh, h^{-1}dh],\ldots,[h^{-1}dh, h^{-1}dh]),$$
for some $C\neq 0.$

The following result is classical \cite[\S4.11]{PrSe}.

\begin{thm}  Let $\{f_i\}$ be a set of generators of the algebra of 
Ad${}_G$-invariant polynomials on $\mathfrak g.$  Then
$\{f_i(\theta^G,[\theta^G, \theta^G],\ldots,[\theta^G, \theta^G]\}$ is a set of ring generators of
$H^*(G,\mathbb{R})$.
\end{thm}

Thus $H^*(G,\R)$ is generated by Chern-Simons classes:

\begin{cor}  For $\{f_i\}$ as above, 
$CS_{f_i}(\nabla_0, \nabla_1)$ is a set of ring
generators of
$H^*(G,\mathbb{R})$
\end{cor} 

For example, for $G = U(n)$ itself, the generators are given by 
$\Tr((h^{-1}dh)^{k})$, although these generators vanish for $k=1$, $k$ even, and $k> n^2-n.$

\subsection{String Chern-Simons forms on loop groups}

Using \cite{PrSe}, we show that generators of the real cohomology ring of a loop group  can be written as Chern-Simons forms for a finite rank bundle.

From now on, $G$ denotes a simply connected compact Lie group. Let $\Omega G$ 
be the group of smooth loops based at the identity. $\Omega G$ in the compact-open topology is an infinite dimensional Lie group with the homotopy type of a CW-complex.  As in (\ref{encap}), 
the evaluation map $\ev:\Omega G\times S^1\to G, \ev(\gamma, \theta) = \gamma(\theta)$, gives
\medskip

$$\begin{CD} 
@.\textrm{ev}^*\underline{V} @>>> \underline{V}\\
@.@VVV   @VVV\\
@.\Omega G\times S^1 @>{\rm ev}>> G\\
@.@V\pi VV @.\\
\calE = \pi_*\ev^*\underline{V}@>>>\Omega G  @.
\end{CD}\ \ \ \ \ \ \ \ \ \ \ \ \ \ \ \ \ \ \ \ \ \ \ \ \ \ $$

\medskip
The $\mathfrak{g}$-valued one-form $\ev^*h_*\theta^G = \ev^*(h^{-1}dh)$
on $ \Omega G\times S^1$ decomposes as
$$\textrm{ev}^*(h^{-1}dh)= \xi+ \eta ,$$
where $\xi$, resp. $\eta$, are supported on $\Omega G$, resp. $S^1$, directions.

It is easy to calculate $\xi$ and $\eta.$

\begin{lem}
(i) At $(\gamma, t)\in \Omega G\times S^1, $
\begin{equation}\label{321} \eta = \gamma^{-1}(t)\dot{\gamma}(t)\,dt.
\end{equation}

(ii) For  $(\gamma_*X,0)\in T_{(\gamma,t_0)}(\Omega G\times S^1)$, we have
\begin{equation}\label{likewise2}
\xi|_{(\gamma,t_0)}(\gamma_*X,0)
=X(t_0).
\end{equation}
\end{lem}

\begin{proof} (i) 
For  $(0,\partial_t)$ tangent to $\{\gamma\}\times S^1$,
$$\eta|_{(\gamma,t_0)}(0,\partial_t)=\textrm{ev}^*(h^{-1}dh)|_{(\gamma,t_0)}(0,\partial_t)=
(h^{-1}dh)|_{\gamma(t_0)}(\textrm{ev}_*(0,\partial_t))
=\gamma(t_0)^{-1}\dot{\gamma}(t_0).$$

(ii) We have
\begin{eqnarray*}
\xi|_{(\gamma,t_0)}(\gamma_*X,0)&=&\textrm{ev}^*(h^{-1}dh)|_{(\gamma,t_0)}(\gamma_*X,0)=
(h^{-1}dh)|_{\gamma(t_0)}(\textrm{ev}_*(\gamma_*X,0))\\
&=&(h^{-1}dh)|_{\gamma(t_0)}(\gamma(t_0)_*X(t_0))=X(t_0).\nonumber
\end{eqnarray*}
\end{proof}

By \cite[\S4.11]{PrSe}, a set of generators for $H^*(\Omega G,\mathbb{R})
= H^{\rm ev}(\Omega G,\R)$ is given by 
\begin{eqnarray}\label{cstilde}
\lefteqn{\int_{S^1}f_i([\xi,\xi],\ldots,[\xi,\xi],\eta)} \\
&=&\int_{S^1}f_i\left(
[\textrm{ev}^*(h^{-1}dh),\textrm{ev}^*(h^{-1}dh)],\ldots,
[\textrm{ev}^*(h^{-1}dh),\textrm{ev}^*(h^{-1}dh)]
,\textrm{ev}^*(h^{-1}dh)\right),\nonumber
\end{eqnarray}
for $\{f_i\}$ a set of generators for the algebra of  Ad-invariant polynomials on $\mathfrak{g}$. To go from the first to the second line in (\ref{cstilde}), we use 
$[\ev^*(h^{-1}dh), \ev^*(h^{-1}dh)] = [\xi+\eta,\xi+\eta] = [\xi, \xi],$
and $\int_{S^1}f_i([\xi, \xi],\ldots, [\xi,\xi], \xi) = 0$.

We want to recognize the right hand side of (\ref{cstilde}) both as a string version of a Chern-Simons form 
and as a contraction of
 a leading order Chern-Simons form.  To begin, we give the Chern-Simons analogues of the the primary string and leading order forms of \S2.

\begin{defn} {\it (i)  Let $Z\to M\stackrel{\pi}{\to}B$ be a fibration of manifolds with $Z$ oriented and closed.  Let $E\to M$ be a vector bundle with structure group $G$ and 
with $G$-connections $\nabla_0, \nabla_1$, and let $\calE = \pi_*E\to B$ be the infinite rank pushdown bundle.  The
string CS form on $\calE$ associated to a degree k invariant polynomial $f$ 
on $\mathfrak g$ is 
$$CS_f^{{\rm str}, \calE}(\pi_*\nabla_0, \pi_*\nabla_1) = \pi_*CS_f^E(\nabla_0, \nabla_1)\in \Lambda^{2k-1-z}(B),$$
where $z = \dim(Z).$

(ii) Assume in addition that $M$ is Riemannian.  The leading order CS form associated to f is
$$CS_f^{{\rm lo}, \calE}(\pi_*\nabla_0, \pi_*\nabla_1) = \int_Z CS_f^E(\nabla_0, \nabla_1)\dvol_Z \in \Lambda^{2k-1}(B),$$
where $\dvol_Z$ is the volume form on the fibers and the integral over Z is in the sense of (\ref{trlo}).  }
\end{defn}

There is a corresponding more natural definition for $G$-principal bundles, which
is the setting for primary string classes in  \cite{h-m-v}, \cite{m-v}.

We see that the generators in (\ref{cstilde}) satisfy
\begin{eqnarray}\label{cstilde2}
\lefteqn{\int_{S^1}f_i\left(
[\textrm{ev}^*(h^{-1}dh),\textrm{ev}^*(h^{-1}dh)],\ldots,
[\textrm{ev}^*(h^{-1}dh),\textrm{ev}^*(h^{-1}dh)]
,\textrm{ev}^*(h^{-1}dh)\right)}\nonumber\\
&=& \pi_*f_i\left(
[\textrm{ev}^*(h^{-1}dh),\textrm{ev}^*(h^{-1}dh)],\ldots,
[\textrm{ev}^*(h^{-1}dh),\textrm{ev}^*(h^{-1}dh)]
,\textrm{ev}^*(h^{-1}dh)\right)\\
&=& \pi_* CS_{f_i}(\ev^*\nabla_0, \ev^*\nabla^1),\nonumber
\end{eqnarray}
where $\nabla_0 = d, \nabla_1 = h^{-1}dh$ are connections on $\underline{V}$ as before.  Of course, $\ev^*\nabla_i$ is a connection on $\ev^*\underline{V}$, so we write
$CS_{f_i}(\ev^*\nabla_0, \ev^*\nabla_1) 
=CS_{f_i}^{\ev^*\underline{V}}(\ev^*\nabla_0, \ev^*\nabla_1).$
  More
explicitly, 
$a=h\circ \ev$ is a gauge transformation on  $\textrm{ev}^*\underline{V}$ 
with  $a^{-1}da=\textrm{ev}^*(h^{-1}dh)$.
Therefore, $\ev^*\nabla_0 = d, \ev^*\nabla_1 = a^{-1}da$,  and
\begin{equation}\label{33}
\pi_*CS_{f_i}^{\ev^*\underline{V}}(\ev^*\nabla_0, \ev^*\nabla_1) = \pi_*
f_i\left([a^{-1}da,a^{-1}da],\ldots,[a^{-1}da,a^{-1}da],a^{-1}da\right).
\end{equation}


Applying the definition of string CS forms to $E=\ev^*\underline{V}, M = \Omega G\times S^1, Z= S^1$ and
using
 (\ref{cstilde}), (\ref{cstilde2}), we obtain

\begin{thm} \label{thmfive}Let $G$ be a compact Lie group, let $\{f_i\}$ be a set of generators for the algebra of  Ad-invariant polynomials on $\mathfrak{g}$, and let
$h':\mathfrak {g}\to \End(V)$ be a faithful finite dimensional representation
with exponentiated representation $h:G\to \Aut(V).$  Take connections
$\nabla_0 = d, \nabla_1 = h^{-1}dh = h^{-1}h'$ on $\underline{V} = G\times V.$
Then  $H^*(\Omega G, \R)$ is generated by 
$$CS_{f_i}^{{\rm str}, \pi_*\ev^*\underline{V}}
( \pi_*\ev^*\nabla_0, \pi_*\ev^*\nabla_1).$$
\end{thm}%


\subsection{Leading order Chern-Simons forms on loop groups}

For the case of circle fibrations $S^1\to M\to B$, there is a relation between the string classes and the leading order classes, both for Chern and Chern-Simons classes.  We will only treat
the CS case for the loop group fibration $S^1\to\Omega G\times S^1\to \Omega G$, but the results immediately extend to loop spaces and 
nontrivial circle fibrations.  In particular, the construction below
produces
nontrivial examples of leading order Chern-Simons forms on loop groups.

Pushing down the trivial bundle $\textrm{ev}^*\underline{V}\to \Omega G\times S^1$
gives a trivial infinite rank bundle $\calV = \pi_*\textrm{ev}^*(\underline{V})\to\Omega G$ with
fiber $C^\infty(S^1,V)$. 
We could also take as fiber a Sobolev completion of $C^\infty(S^1,V)$ to produce a
Hilbert bundle. 

Associated to the gauge transformation $h$ of $\underline V$ is the gauge transformation
$$\tilde{h}:\Omega G\to \textrm{Aut}(\calV),$$
given by
$\tilde{h}_\gamma (s)(t)=h_{\gamma(t)}(s(t))$ for $s\in C^\infty(S^1,V)$.
By abuse of notation, $\tilde{h}^{-1}d\tilde{h}\in \Lambda^1(\Omega G,\Omega\mathfrak{g}) =
\Lambda^1(\Omega G, C^\infty(S^1)\otimes \mathfrak g)$ can be identified with the Maurer-Cartan form $\theta^{\Omega G}$ on $\Omega G$,
and we have connections $\widetilde{\nabla}_0=d, \widetilde{\nabla}_1=\tilde{h}\cdot\widetilde{\nabla}_0 = \tilde h^{-1} d\tilde h$ on $\calV$.
For $f_i$ an ${\rm Ad}_G$-invariant polymonial on $\mathfrak g$, 
\begin{align}\label{inf}
CS^{\Omega G,\infty}_{f_i}(\widetilde{\nabla}_0, \widetilde{\nabla}_1)
&\stackrel{\rm def}{=} f_i(\theta^{\Omega G},[\theta^{\Omega G},\theta^{\Omega G}]\ldots,[\theta^{\Omega G},\theta^{\Omega G}])\\
&=f_i(\tilde{h}^{-1}d\tilde{h},[\tilde{h}^{-1}d\tilde{h},\tilde{h}^{-1}d\tilde{h}]
\ldots,[\tilde{h}^{-1}d\tilde{h},\tilde{h}^{-1}d\tilde{h}])\nonumber
\end{align}
belongs to $\Lambda^*(\Omega G,C^\infty(S^1))$,
 because $f_i$ only acts on the $\mathfrak g$ part of $C^\infty(S^1)\otimes \mathfrak g$.

The averaging map in (\ref{trlo}) in our context is 
$\Upsilon:\Lambda^*(\Omega G, C^\infty(S^1))\to \Lambda^*(\Omega G,\mathbb{C})$ given by
 $$\Upsilon(\omega) = \int_{S^1}\omega\wedge dt.$$
 Note that $\Upsilon$ does not lower the degree of $\omega.$  Let 
 $ d_{\Omega G}$ be the exterior derivative on $\Lambda^*(\Omega G)$, and let
  $d_{\Omega G}^\nabla$ be the exterior derivative coupled to the trivial connection on 
$\Omega G\times  C^\infty(S^1)\to\Omega G$. Then $\Upsilon d_{\Omega G}^\nabla = d_{\Omega G} \Upsilon,$ so
$\Upsilon$ induces a map on cohomology groups.  In particular, 
$\Upsilon(CS^{\Omega G,\infty}_{f_i}(\widetilde{\nabla}_0, \widetilde{\nabla}_1))$
 is precisely the leading order Chern-Simons class $CS^{{\rm lo}, \calV}_{f_i}
(\widetilde \nabla_0, \widetilde\nabla_1)$ on $\calV.$

Let $\chi(\gamma) = \dot\gamma$ be the fundamental vector field on $\Omega G$
with associated Lie algebra valued function $\chi(\gamma) = \theta^{\Omega G}(\dot \gamma)
=\gamma_*^{-1}\dot{\gamma}\in\Omega\mathfrak{g}$.  Note that $\Upsilon \iota_\chi = \iota_\chi\Upsilon$, because
$\iota_\chi$ involves $\Omega G$ variables and $\Upsilon$ integrates out the $S^1$ information.

We can now relate the string and leading order CS classes, and prove that the 
contraction of the leading order CS classes with $\chi$ generate $H^*(\Omega G).$

\begin{thm} \label{relate} 
Let $f$ be an ${\rm Ad}_G$-invariant polynomial on $\mathfrak g.$  Then for the connections $\widetilde{\nabla}_0, \widetilde{\nabla}_1$, we have
$$\iota_\chi (CS^{{\rm lo}, \calV}_{f} (\widetilde{\nabla}_0, \widetilde{\nabla}_1))
= CS^{{\rm str}, \calV}_{f}(\widetilde{\nabla}_0, \widetilde{\nabla}_1).$$
In particular, if
 $\{f_i\}$ generate the algebra of invariant polynomials on $\mathfrak g,$ then
closed forms
$$\iota_\chi (CS^{{\rm lo}, \calV}_{f_i}(\widetilde{\nabla}_0, \widetilde{\nabla}_1))
  = \iota_\chi\Upsilon(CS^{\Omega G,\infty}_{f_i}(\widetilde{\nabla}_0, 
  \widetilde{\nabla}_1))$$
generate $H^*(\Omega G)$.
\end{thm}


\begin{proof} 
We have
$$\iota_\chi (CS^{{\rm lo}, \calV}_{f_i}(\widetilde{\nabla}_0, \widetilde{\nabla}_1))
= \iota_\chi \Upsilon(CS^{\Omega G,\infty}_{f_i}(\widetilde{\nabla}_0, \widetilde{\nabla}_1)) = \Upsilon\iota_\chi(CS^{\Omega G,\infty}_{f_i}(\widetilde{\nabla}_0, \widetilde{\nabla}_1)).$$
Take 
$\gamma_*X_1,\ldots,\gamma_*X_{2k-2}\in T_\gamma\Omega G$, where $X_j\in \Omega\mathfrak{g}$  and $k=\textrm{deg}(f)$. 
Then for  permutations $\sigma\in \Sigma_{2k-2}$,
 \begin{eqnarray*}\lefteqn{\Upsilon\iota_{\chi}( CS^{\Omega G,\infty}_{f}
 (\widetilde{\nabla}_0, \widetilde{\nabla}_1))
 (\gamma_*X_1,\ldots,\gamma_*X_{2k-2}) }\\
&=&
\sum_{\sigma}(-1)^{|\sigma |}\int_{S^1}
 f\left([X_{\sigma(1)},X_{\sigma(2)}],\ldots,
[X_{\sigma(k-3)},X_{\sigma(2k-2)}],\gamma_*^{-1}(t)\gamma(t)\right)dt\\
&=& \left(
\int_{S^1}
f([\xi,\xi],\ldots,
[\xi,\xi],\eta)\right)(\gamma_*X_1,\ldots,\gamma_*X_{2k-2}) \\
&=&  CS^{{\rm str}, \calV}_{f}(\widetilde{\nabla}_0, \widetilde{\nabla}_1)
 (\gamma_*X_1,\ldots,\gamma_*X_{2k-2}), 
\end{eqnarray*}
where we use (\ref{321}), (\ref{likewise2}), (\ref{inf}).
\end{proof}

As an example, the even cohomology classes of the string CS forms $\int_{S^1}\tr(\iota_\chi(\tilde h^{-1}d\tilde h)^{2k-1})$ generate\\
 $H^*(\Omega U(n),\R)$, whereas the odd cohomology classes of the leading order 
CS forms 
$\int_{S^1} \tr((\tilde h^{-1}d\tilde h)^{2k-1})$ vanish.

In summary, both string and leading order Chern-Simons forms give representatives of the generators of
$H^*(\Omega G, \R).$  The use of  string CS classes is more natural, reflecting the fact that the 
primary string classes are topological objects.  In fact, the relation between the string CS classes and the
contracted leading order CS classes appears only because we have an $S^1$-fibration.

\section{Leading order classes and currents in gauge theory}  

Let $P\to M$ be a principal $G$-bundle over a closed manifold $M$ with compact semisimple group $G$. We denote by 
$\calas$, resp. $\calG$,  the space of irreducible connections on $P$, resp. 
 the gauge group of $P$. In this section we show that the leading order Chern classes of the canonical connection on the principal gauge bundle
$\calas\to\ \calas/\calG = \calB^*$ are related to Donaldson classes.

We put appropriate Sobolev norms on $\calas$ and $\calG$, so that the moduli space
$\calB^* = \calas/\calG$ is a Hilbert manifold.  The right action of $\calG$ on $\calA$ is the usual
$A\cdot g = {\rm Ad}_g(A),$ recalling that $A\in \Lambda^1(P,\frakg)$ with the
adjoint action  on $\frakg$.  Set $\adP = P\times_{\rm Ad}\frakg.$
 The tangent space $T_A\calA$ is canonically isomorphic to 
$\Lambda^1(M, \adP)$.  
A fixed metric $h= (h_{ij})$ on $M$ induces a Riemannian or $L^2$ metric on 
$T\calA$ by
$$\langle X, Y\rangle_1 = \int_M h^{ij}\langle A_i, B_j\rangle \dvol_h,$$
where $X = A_i dx^i, Y = B_jdx^j\in T_A\calA$ and $\langle\ ,\ \rangle$ is an Ad${}_G$-invariant  positive definite inner product on $\adP$.  
 Since the derivative of the gauge action  (also denoted by 
$\cdot g$) is $X\cdot g = {\rm Ad}_g(X)$, the metric is gauge invariant. 

The Lie algebra Lie$(\calG)$ of $\calG$ is $\Lambda^0(M, \adP),$ which has the
$L^2$ metric $\langle f,g\rangle_0 = \int_M \langle f,g\rangle \dvol_h.$
 Let $d_A:
{\rm Lie}(\calG) = \Lambda^0(M, \adP)\to \Lambda^1(M, \adP)$ be the covariant derivative associated to $A$.  Then the vertical 
space of  $\calas\to \calB^*$ at $A$ is Im$(d_A)$.  It is straightforward to check that the orthogonal complement 
$\keras$ forms the horizontal space of a connection
on $\calas\to\calB^*.$ Let $\omega$ be the corresponding connection one-form. 

Let $\Omega$ be the curvature of $\omega.$  $\Omega$ is horizontal.  An explicit formula for $\Omega$ is known \cite{GP, singer}, but  to our knowledge
the following proof has not appeared.

\begin{lem} \label{appeared} For $X, Y$ horizontal tangent vectors at $A$, we have 
$$\Omega(X,Y) = -2 G_A *[X,*Y]\in {\rm Lie}(\calG).$$
\end{lem}

Here  $\Delta_A = d_A^* d_A, G_A = \Delta_A^{-1}$ is the Green's operator for $\Delta_A$, and $*$ is the Hodge star associated to $h.$

\begin{proof}
For $X\in T_A\calas$, let $X^h, X^v$ denote the horizontal and vertical components of $X$.  
As a vertical vector, the curvature of $\omega$ at $A$ is 
\begin{eqnarray*} \Omega(X,Y) &=& d\omega(X^h, Y^h) = X^h(\omega(Y^h)) - 
Y^h(\omega(X^h)) -
\omega([X^h, Y^h])\\
&=&  -\omega([X^h, Y^h]) = -[X^h, Y^h]^v,
\end{eqnarray*}
for any 
 extension of $X, Y$ to vector fields near $A$. 
  We have 
$$X^v = d_AG_Ad_A^* X,\ \  X^h = X - d_AG_Ad_A^* X.$$
In a local trivialization of $\calas\to\calB^*$, we can write $[X^h, Y^h] = 
\delta_{X^h} Y^h -\delta_{Y^h} X^h.$ 

We may extend $X, Y$ to constant vector fields near $A$ with respect to this trivialization, so for any tangent vector $Z$, $\delta_ZY =\delta_ZX=0$ at $A$.  Then
\begin{eqnarray*} \delta_{X^h} Y^h &=& \delta_{X-\daga X}(Y-\daga Y) = 
-  \delta_{X-\daga X}\daga Y\\
&=& -(\delta_X d_A) G_Ad_A^*Y - d_A(\delta_X G_A)d_A^* Y 
- d_A G_A (\delta_X d_A^*)Y\\
&&\qquad  + (\delta_{\daga X}d_A)G_Ad_A^* Y + d_A(\delta_{\daga X}G_A)d_A^*Y
+ d_AG_A(\delta_{\daga X} d_A^*)Y\\
&=& - d_A G_A (\delta_X d_A^*)Y +  d_AG_A(\delta_{\daga X} d_A^*)Y,
\end{eqnarray*}
since $d_A^*Y = 0$ at $A$.   Locally, $d_A^* = - *d_A* = *(d+ [A,*\cdot])$, so
$$\delta_X d_A^*  = -(d/dt)|_{t=0} *(d+ [A+tX,*\cdot]) = -*[X,*\cdot].$$
Thus
$$\delta_{X^h}Y^h = d_AG_A*[X,*Y] -d_AG_A[\daga X,*Y]= d_AG_A*[X,*Y],$$
since $d_A^*X = 0.$

For $X = A_ie^i, Y = B_je^j$ in a local orthonormal frame $\{e^i\}$, 
$$[X,*Y] = [A_i, B_j]e^i\wedge *e^j  = \sum_i [A_i, B_i] e^i\wedge *e^i
= -\sum_i [B_i, A_i]e^i\wedge *e^i = -[Y, *X].$$  
Therefore $\delta_{X^h} Y^h -\delta_{Y^h}X^h = 2d_AG_A*[X,*Y].$ This gives
$$ \Omega(X, Y) = -\daga(\delta_{X^h} Y^h -\delta_{Y^h}X^h) 
= -2\daga d_AG_A*[X,*Y] = -2 d_AG_A*[X,*Y]$$
as a vertical vector.  Since $d_A^{-1}$ is takes vertical vectors isomorphically to Lie$(\calG)$, we get
$$\Omega(X,Y) = -2G_A*[X,*Y]\in {\rm Lie}(\calG).$$

\end{proof}

The curvature takes values in Lie$(\calG) = \Lambda^0(M, {\rm Ad}\  P)$. Up until now, we have only 
considered connections on vector bundles, where the curvature takes values in an endomorphism bundle.  
If $G$ is a matrix group, Lie$(\calG)$ has a global trace given by integrating the trace on $\frakg,$ the fiber of \\
Ad $P$. In general, Lie$(\calG)$ can be thought of as an algebra of multiplication
operators
 via the injective adjoint  representation of
$\frakg$. Equivalently, we can pass to the vector bundle
${\rm Ad}\ \calas = \calas\times_{\calG} {\rm Lie}(\calG)$ with fiber ${\rm Lie}(\calG)$ and 
take the leading order classes of its associated connection $d{\rm Ad}(\omega)$, whose curvature   $[\Omega, \cdot]$ 
is usually  denoted just by $\Omega.$
  Either way, the leading order Chern form $\cklo(\Omega)$ of $\calas\to\calB^*$ is
a positive multiple of
$$\int_M \tr(\Omega^k)\dvol_h.$$

We claim that if
$\dim(M) = 4$, we can
 identify $\ctlo(\Omega)$ with Donaldson's $\nu$-form.   We briefly recall the 
construction of this form \cite[Ch. 5]{dk}.  Let $\widetilde P = \pi^* P\to \calas\times
M$ be the pullback bundle for the projection $\pi:\calas\times M\to M.$  $\widetilde P=\calas 
\times P$ has
the connection $A$ on the slice $\{A\}\times M$.
$\widetilde P$ 
descends to a $G^{\rm ad} = G/Z(G)$-bundle, denoted $\Pad$, over $\calB^*\times M.$  
$\Pad\to\calB^*\times M$ has a family of framed connections, denoted $(A,\phi)$, once a framing
is fixed at some $m_0\in M.$  
For example, if $G = SU(2)$, then $\Pad$ is a $SO(3)$-bundle.  $\nu$ is defined by
$$\nu = -\frac{1}{4} p_1(\Pad).$$
 By the calculation in \cite[\S5.2]{dk}, 
the form  $\nu = p_1^{\rm lo}(\calas)= \ctlo(({\rm Ad} \ \calas)\otimes \C)$ is given by
$$\nu(X_1, X_2, X_3, X_4) = c\cdot\sum_{\sigma\in \Sigma_4}{\rm sgn}(\sigma) \int_M \tr(G_A*[X_{\sigma(1)},
 *X_{\sigma(2)}] G_A*[X_{\sigma(3)}, *X_{\sigma(4)}]) \dvol,$$
 for some constant $c$.
By Lemma \ref{appeared}, we obtain.

\begin{prop} \label{nvs} As differential forms, $\nu$ equals $p_1^{\rm lo}(\calas)$ up to a constant. 
\end{prop}

It is of course more interesting to relate leading order classes to $\mu(a)$ for 
$a\in H_2(M,\Q)$, and  Donaldson's map
$\mu:H_*(M,\Q)\to H^{4-*}(\calM,\Q)$. ($\calM$ is the moduli space 
 of ASD
connections.)  Recall that $\mu(a) = i^*(\nu/a)$, for the slant product $\nu/:
H_*(M,\Q)\to H^{4-*}(\calB^*,\Q)$ and $i:\calM\to\calB^*$ the inclusion.  In particular,
$\nu = \mu(1)$ for
$1\in H_0(M).$  The difficulty in showing that $\calM$ has a fundamental class, so that the
slant product is defined, is similar in spirit to the issues treated in \S4.

There is a positive integer $k$ such that $ka$ has a representative 
submanifold, and therefore $\PD(ka)$ has a representative closed two-form $\omega.$
If $k\neq 1$, 
we replace $\omega$ by $k^{-1}\omega.$
In general, let $\omega$ be a closed two-form on $M$. By  \cite[Prop. 5.2.18]{dk}, the two-form $C_\omega\in \Lambda^2(\calM)$ representing $\nu/a = \nu/\PD^{-1}(\omega)$ and 
hence
$\mu(a) $ is given at
$[A]\in \calM$ by
\begin{equation}\label{mu}C_\omega(X,Y) = \frac{1}{8\pi^2}\int_M \tr(X\wedge Y)\wedge \omega
+ \frac{1}{2\pi^2}\int_M \tr(\Omega_A(X,Y) F_A)\wedge \omega,
\end{equation} 
where $F_A$ is the curvature of $A$.  On the right hand side, we use any $A\in [A]$ and 
$X, Y\in T_A\calA^*$ with  $d_A^*X = d_A^*Y = 0.$

There is a leading order class associated to any distribution or zero current $\Lambda$ on $C^\infty(M)$, given pointwise by
$$c_k^{{\rm lo}, \Lambda} = \Lambda (\tr(\Omega^k)),$$
where $\Omega$ is the curvature of a connection taking values in the Lie algebra of 
a gauge group, as in this section.
(More generally, there are leading order classes associated to distributions on the unit 
cosphere bundle of $M$
for connections taking values in 
nonpositive order pseudodifferential operators \cite{L-N}, \cite{P-R2}.)
In particular, for a fixed $f\in C^\infty(M)$ we have the characteristic class
$$\int_M f\cdot \tr(\Omega^k).$$
   We can turn this around and consider $\tr(\Omega^k)$ as a zero-current acting on $f$. Looking back at (\ref{mu}), we can consider the two-currents  
\begin{equation}\label{curr}
\tr(X\wedge Y), 
\ \tr(\Omega_A(X,Y) F_A),
\end{equation}
for fixed $X, Y$.  Thus we can consider $C$ as an element of 
$\Lambda^2(\calM, \calD^2)$, the space of two-current valued two-forms on $\calM.$

Because these two-currents are Ad${}_{\calG}$-invariant, the usual Chern-Weil proof shows that $C(\omega) = C_\omega$ is closed.  (Its class is of course independent of the connection on 
$\calA^*$, but we have a preferred connection.) $C$ is built from Ad${}_G$-invariant functions but only the first term in (\ref{curr}) comes from an invariant polynomial in 
${\rm Lie}(\calG)^\calG.$ Nevertheless, we 
interpret  (\ref{mu}) as a sum of
``leading order currents" evaluated on $\omega.$  

\begin{prop} For $a\in H_2(M^4,\Q)$,
a representative two-form for  Donaldson's $\mu$-invariant $\mu(a)$ is given by
evaluating the leading order two-current 
$$ \frac{1}{8\pi^2}\int_M \tr(X\wedge Y)\wedge \cdot
+ \frac{1}{2\pi^2}\int_M \tr(\Omega_A(X,Y) F_A)\wedge \cdot $$
on any two-form Poincar\'e dual to a.
\end{prop}

 \bibliographystyle{amsplain}
\bibliography{Paper3}

\end{document}